\documentclass{svproc}
 
\usepackage{url}

\usepackage{algorithm}
\usepackage{amssymb}

\usepackage[font=small]{caption}
\usepackage[font=footnotesize]{subcaption}
\usepackage{booktabs}
\usepackage{multirow}

\usepackage{tikz}
\usetikzlibrary{automata,positioning}
\usepackage{tkz-berge}

\usepackage{epstopdf}
\usepackage{graphicx,url}
\usepackage{epstopdf}
\usepackage{color}
\usepackage{float}
\usepackage{comment}
\usepackage{mathtools}

\DeclareMathOperator*{\argmin}{arg\,min}

\usepackage{makecell}
\usepackage{algpseudocode}

\newcommand\oprocendsymbol{\hbox{$\bullet$}}
\newcommand\oprocend{\relax\ifmmode\else\unskip\hfill\fi\oprocendsymbol}

\urldef{\smith}\url{stephen.smith@uwaterloo.ca}
\urldef{\armin}\url{a6sadegh@uwaterloo.ca}
\urldef{\ahmad}\url{abasghar@uwaterloo.ca}

\algtext*{EndWhile}% Remove "end while" text
\algtext*{EndIf}% Remove "end if" text
\algtext*{EndFor}% Remove "end while

\title{\Large \textbf{Distributed Multi-Robot Coverage Control of Non-convex Environments with Guarantees}}
\title{\Large \textbf{Approximation Algorithms for Distributed Multi-Robot Coverage in Non-Convex Environments}}
\titlerunning{Distributed Coverage of Non-convex Environments}

\author{ Armin Sadeghi, Ahmad Bilal Asghar and Stephen L. Smith}

\authorrunning{A.Sadeghi, A.B.Asghar, S.L.Smith} 

\tocauthor{Armin Sadeghi \qquad Ahmad Bilal Asghar \qquad Stephen L. Smith}

\institute{Department of Electrical and Computer Engineering\\
University of Waterloo, Waterloo, ON, Canada\\
\texttt{\{a6sadegh,abasghar,stephen.smith\}@uwaterloo.ca}}

\begin{document}
\maketitle
\begin{abstract}
	In this paper, we revisit the distributed coverage control problem with multiple robots on both metric graphs and in non-convex continuous environments. Traditionally, the solutions provided for this problem converge to a locally optimal solution with no guarantees on the quality of the solution. We consider sub-additive sensing functions, which capture the scenarios where sensing an event requires the robot to visit the event location. For these sensing functions, we provide the first constant factor approximation algorithms for the distributed coverage problem. The approximation results require twice the conventional communication range in the existing coverage algorithms. However, we show through extensive simulation results that the proposed approximation algorithms outperform several existing algorithms in convex, non-convex continuous, and discrete environments even with the conventional communication ranges. Moreover, the proposed algorithms match the state-of-the-art centralized algorithms in the solution quality. 
	 
\end{abstract}

\section{Introduction}
\label{sec:intro}

Distributed Coverage is a very well studied problem~\cite{cortes2004coverage,durham2012discrete,sadeghi2019coverage,santos2018coverage,wang2011coverage} with extensive multi-robot applications, such as environmental monitoring~\cite{curtin1993autonomous}, and  surveillance~\cite{meguerdichian2001exposure}. The objective is to deploy a set of robots to cover an environment such that each robot services or senses the events closer to that robot than any other robot. The events arrive according to a spatial distribution, and the cost of sensing an event is a function of the distance from the robot to that event. The distributed coverage control problem is to minimize the total coverage cost of the environment. The existing distributed algorithms to solve this problem converge to a locally optimal solution with no guarantees on the quality of the solution. In this paper, we provide distributed approximation algorithms to solve the problem in non-convex continuous and discrete environments. 

The first distributed algorithm for coverage control in convex environments was proposed by Cortes \emph{et al.}~\cite{cortes2004coverage}. The algorithm utilizes Lloyd's descent to converge to a locally optimal solution, and the partition of each robot is defined using Voronoi partitioning. The robots communicate with the robots in their neighboring partitions to implement the algorithm.  

Building on the Lloyd's descent-based algorithm in~\cite{cortes2004coverage}, there has been extensive studies on the coverage control problem in non-convex environments. In~\cite{caicedo2008performing,caicedo2008coverage}, the authors map non-convex environments through a diffeomorphism to a convex region and then solve the problem using the Lloyd's algorithm~\cite{cortes2004coverage} before mapping the locally optimal solution back to the original environment. 

For non-convex polygonal environments, a distributed algorithm was presented in~\cite{breitenmoser2010voronoi}, where the Lloyd's algorithm for convex environments was combined with a local path planning algorithm to avoid obstacles. A key idea introduced in~\cite{thanou2013distributed} is to consider geodesic distance when computing partitions. In~\cite{kantaros2014visibility} and~\cite{mahboubi2012distributed}, the authors construct the Voronoi partitions based on the visibility of the robot in the presence of obstacles. Unlike the approaches mentioned above, we discretize the non-convex environment, and solve the coverage problem on the discrete environment and provide guarantees on the solution quality.

The approach of converting a continuous non-convex environment to a discrete environment is used in~\cite{durham2012discrete,alitappeh2017multi,bhattacharya2013distributed}. We utilize the same approach, but are able to characterize the cost of the solution obtained from the discretized environment in the corresponding continuous environment as a function of the sampling density. The authors in~\cite{yun2014distributed} study the coverage problem defined on an undirected graph and present a distributed algorithm that converges to a local optimum. Their algorithm requires the robots to know the information of the neighbors of their neighbors. In this paper, we make the same assumption on the communication range of the robots, but establish approximation guarantees. 

A closely related problem to the discretized coverage control is the facility location problem~\cite{jain2001approximation,shmoys2000approximation} where the objective is to minimize the cost of the robots and the total service time of the demands arriving on the vertices. A special case of this problem is the $k$-median problem~\cite{arya2004local,li2016approximating,ahmadian2013local} where the objective is to place $k$ robots on vertices of the graph to minimize the total service time. A centralized approximation algorithm was presented for the $k$-median problem in~\cite{arya2004local}, and the analysis of our distributed approximation algorithm leverages this centralized approximation algorithm. The authors in~\cite{ahmadian2013local} consider the $k$-median problem with mobile robots, namely the mobile facility location problem, and provide an approximation algorithm for the objective of minimizing a linear combination of the relocation cost of the mobile robots and the expected service time of the demands. In~\cite{sadeghi2018re}, we consider the mobile facility location problem with sequentially arriving demands where the goal is to minimize a linear combination of the relocation costs and the expected service times of the demands in a time horizon, and propose a centralized algorithm which provides solutions within a constant factor of the optimal solution.
Authors in~\cite{balcan2013distributed} provide a randomized-distributed algorithm for the $k$-median problem with constant factor approximation in Euclidean environments. In contrast, we consider more general non-convex environments and provide a deterministic approximation algorithm.

\emph{Contributions:} Our main contributions are threefold.  First, given a continuous non-convex coverage problem, we generate a corresponding instance on a metric graph, and characterize the performance of the discrete solution on the continuous problem (Section~\ref{sec:discrete-problem}). Second, we provide a constant factor approximation algorithm for the distributed coverage problem on metric graphs (Section~\ref{sec:algorithm}). To the best of our knowledge, this is the first deterministic approximation algorithm for the distributed multi-robot coverage problem. We prove the approximation results in Section~\ref{sec:analysis}. Third, we show through extensive simulations that the proposed algorithm outperforms several existing approaches in convex and non-convex environments, and matches the centralized algorithms in solution quality (Section~\ref{sec:simulations}).

\section{Continuous and Discrete Coverage Problems}
\label{sec:problem_formulation}
We begin by reviewing the coverage problems in both continuous~\cite{cortes2004coverage} and discrete environments~\cite{yun2014distributed}.

\subsection{Continuous Environment}
Consider $m$ mobile robots in a compact environment with obstacles and let $\mathcal{X}$ be the obstacle free subset of the environment.  There is an event distribution $\phi: \mathcal{X} \rightarrow \mathbb{R}_+$ defined over the environment. Let $d(p, q)$ be the length of the shortest path between two locations $p$ and $q$ in $\mathcal{X}$. The sensing cost of an event at location $p$ by a robot at $q$ is a strictly non-decreasing function $f:\mathbb{R}_+ \rightarrow \mathbb{R}_+$ of $d(p, q)$. Following the non-convex problem formulation in~\cite{breitenmoser2010voronoi}, which extends the original formulation in~\cite{cortes2004coverage}, the continuous problem is defined as the problem of finding the set of locations in the environment for Authorsthe robots that minimizes the sensing cost of the events, i.e.,
\begin{equation}
	\label{eq:cont_coverage_objective}
	\min_{Q\in \mathcal{X}^m}\mathcal{H(Q)} = \min_{Q \in \mathcal{X}^m}\int_{\mathcal{X}} \min_{q_i \in Q}f(d(p,q_i))\phi(p)dp.  
\end{equation}

Without loss of generality, in the rest of the paper we assume that $\int_{\mathcal{X}} \phi(p) dp = 1$. Observe that the best sensing cost for an event is provided by the closest robot to that event location. Then for a given configuration $Q$, we partition the environment into Voronoi subsets as follows:
\[
	V_i(Q) = \{p \in \mathcal{X} | d(p, q_i) \leq d(p, q_j) \ \forall q_j \in Q \setminus \{q_i\}\}.
\]

The robots move according to some dynamics $\dot{q}_i = g(q_i, u_i)$ where the computation of the shortest path between two configurations of the robot is tractable. Typically first order dynamics $g(r_i, u_i) = u_i$ is considered for the robots in coverage control literature~\cite{cortes2004coverage}. 
We are interested in the distributed version of the coverage problem, where the robots have local information on the other robots and each robot computes its control input locally.

\subsection{Discrete Environment}
\label{subsec:discrete_env_problem_def}

Consider a metric graph $G= (\mathcal{S}, E, c)$ where $\mathcal{S}$ is the vertex set, $E$ is the set of edges between the vertices, $c$ is the metric edge cost and let $w$ be the weight on the vertices. Given a team of $m$ robots, the coverage problem on graph $G$, is the problem of finding a set of $m$ locations to optimally cover the vertices of $G$, i.e., minimize $\mathcal{D}(Q) = \sum_{v \in \mathcal{S}} \min_{q \in Q} w(v) c(v, q)$.
For a given configuration $Q \in \mathcal{S}^m$, we can partition the vertices into $m$ subsets 
\[W_i(Q) = \{u \in \mathcal{S} | c(u, q_i) < c(u, q_j) \ \forall q_j \in Q\setminus \{q_i\}\}.\]
    
If there exists a vertex that has equal distance to two or more robots in $Q$, then the vertex is assigned to the robot with smaller unique identifier (UID). Robots travel on the edges of the graph and the control input to a robot is a sequence of edges leading to its destination vertex. 

\emph{Centralized Approximation Algorithm:} The centralized version of this problem is a well-known NP-hard problem called the $k$-median problem~\cite{arya2004local}. The best known approximation algorithm for this problem on metric graphs is a centralized local search algorithm which provides solutions within a constant factor of the optimal~\cite{ahmadian2013local}. Starting from a configuration $Q$, the centralized local search algorithm swaps $p$ vertices in $Q$ at a time with a subset of $p$ vertices in $\mathcal{S} \setminus Q$. If the new configuration improves the coverage by at least some $\epsilon_0 > 0$, then we call this move a valid local move. The procedure terminates if there are no more valid swaps improving the total sensing cost. We will refer to this local search algorithm as $\textsc{CentralizedAlg}$ in the rest of the paper. The solution obtained from $\textsc{CentralizedAlg}$ is within $3 + 2/p + o(\epsilon_0)$ of the globally optimal solution. 

We focus on the distributed version of this problem introduced in~\cite{yun2014distributed} where the robots use only the local information to compute their control input. In the following section, we establish the connection between the continuous and discrete coverage problems.

\section{From a Continuous to a Discrete Problem}
\label{sec:discrete-problem}
To establish a connection between the coverage problem in continuous and discrete environments, we first convert the continuous coverage problem to a coverage problem in a discrete environment through sampling of the environment. 

Let $\mathcal{S}$ be the set of samples of $\mathcal{X}$ with dispersion $\zeta$~\cite{lavalle2006planning}, where $\zeta$ is the maximum distance of any point in the environment $\mathcal{X}$ from the closest point in $\mathcal{S}$, i.e.,
$
\zeta = \max_{p \in \mathcal{X}} \min_{u \in \mathcal{S}} d(p, u),
$
(See Figure~\ref{fig:discretization}).
We construct a metric graph $G = (\mathcal{S}, E, c)$ on sampled locations $\mathcal{S}$, where $E$ is the edge set and $c$ is a function assigning costs to the edges of the graph. The cost of an edge between two sampled locations $u, v \in \mathcal{S}$ is $c(u, v) = f(d(u, v))$. Let $\sigma(v)$ for $v \in \mathcal{S}$ be the points in $\mathcal{X}$ closer to $v$ than other samples in $\mathcal{S}$, i.e., 
\[\sigma(v) = \{p\in \mathcal{X}| d(p, v) \leq d(p, u) \ \forall u \in \mathcal{S}\setminus \{v\}\}.\]
With a slight abuse of notation, let $\sigma^{-1}(p)$  be the closest sample in $\mathcal{S}$ to $p \in \mathcal{X}$. The function $w: \mathcal{S} \rightarrow \mathbb{R}_{+}$ assigning weights to the vertices of the graph is $w(v) = \int_{p \in \sigma(v)} \phi(p)dp$. We assume the following property on the sensing function.
\begin{figure}[t]
	\centering
	\includegraphics[angle=-90,width=0.45\textwidth]{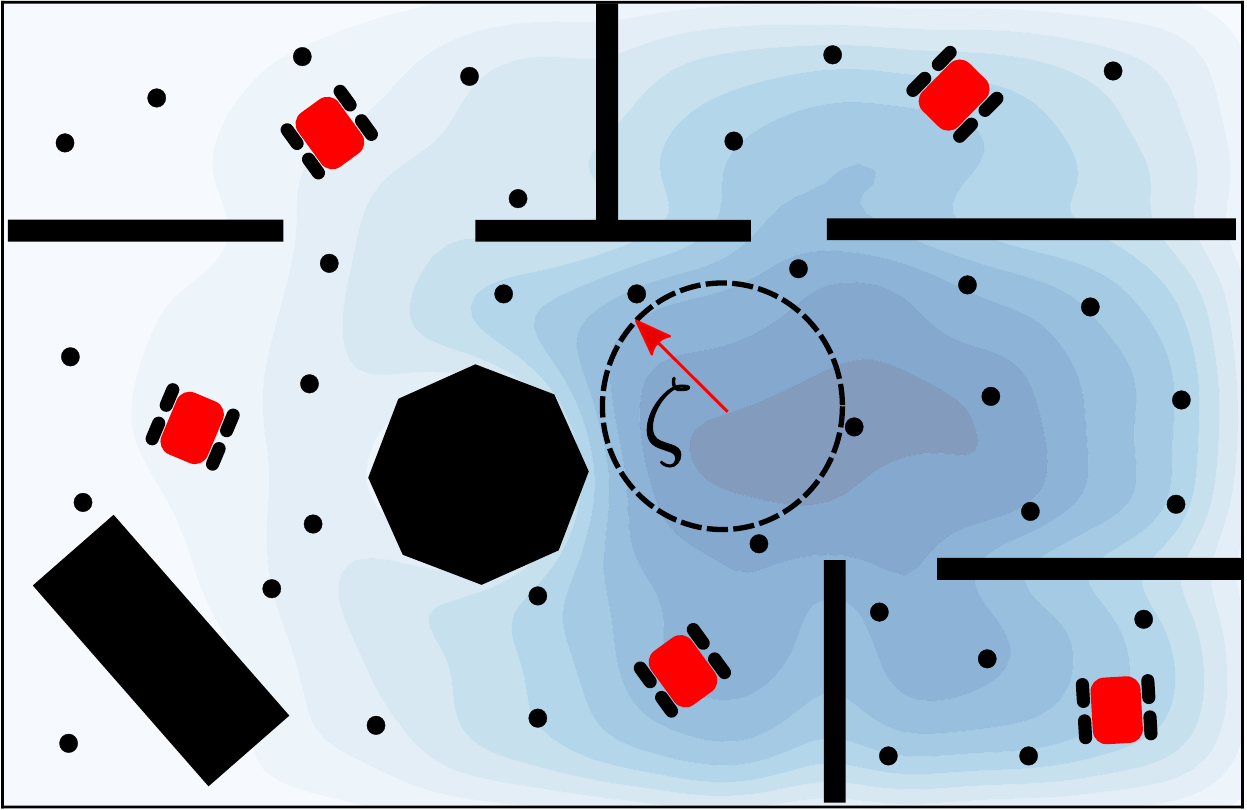}
	\caption{Sampled locations in an environment with dispersion $\zeta$.}
	\label{fig:discretization}
\end{figure}
\begin{assumption}[Subadditivity of sensing function]
	\label{assump:metric-sensing}
	We assume that the sensing cost function $f$ is a sub-additive function, i.e., \[f(d(p, u) + d(u, v)) \leq f(d(p, u)) + f(d(u, v)).\] 
\end{assumption}
For instance $f(x) = \sqrt{x}$ and $f(x) = x$ are  sub-additive functions.
\begin{remark}
	In applications such as dynamic vehicle routing problems~\cite{bertsimas1991stochastic,bullo2011dynamic} and facility location problems~\cite{shmoys2000approximation} where the sensing cost is determined by the distance traveled by the clients, the sensing function falls under Assumption~\ref{assump:metric-sensing}. 
\end{remark}

Due to Assumption~\ref{assump:metric-sensing}, the cost function $c$ on the edges of the graph $G$ satisfies the triangle inequality, i.e., for all $u, v, z$ in $\mathcal{S}$
\begin{align*}
	c(u, v) & = f(d(u, v)) \leq f(d(u, z) + d(z, v))            \\
	        & \leq f(d(u, z)) + f(d(z, v)) = c(u, z) + c(z, v). 
\end{align*}

The following result establishes a connection between the sensing costs of an approximate solution to the discrete coverage problem and the optimal coverage in continuous environment.

\begin{theorem}
	\label{thm:proof-cont-approx}
	Consider a continuous coverage problem on environment $\mathcal{X}$ with optimal solution $S^*$, and its corresponding discrete instance obtained through the set of samples $\mathcal{S}$ with dispersion $\zeta$. Then if $Q$ is the solution obtained from an $\alpha$-approximation algorithm for the discrete coverage problem instance,  the sensing cost of $Q$ on the corresponding continuous problem is $\mathcal{H}(Q) \leq \alpha \mathcal{H}(S^*) + o(f(\zeta))$.
\end{theorem}
\begin{proof}
	We have,
	\begin{align}
		\mathcal{H}(Q) & = \int_{\mathcal{X}} \min_{q_i \in Q} f(d(q_i, p)) \phi(p) dp\nonumber                                                                                                  \\
		               & \leq \int_{\mathcal{X}} \min_{q_i \in Q} f(d(q_i, \sigma^{-1}(p)) + d(\sigma^{-1}(p), p))) \phi(p) dp\quad \text{(triangle inequality)}\nonumber                        \\
		               & \leq \int_{\mathcal{X}} \min_{q_i \in Q} [f(d(q_i, \sigma^{-1}(p)))+ f(d(\sigma^{-1}(p), p))] \phi(p) dp \quad \text{(Assumption~\ref{assump:metric-sensing})}\nonumber \\
		               & = \int_{\mathcal{X}} \min_{q_i \in Q} f(d(q_i, \sigma^{-1}(p)))\phi(p) dp+  \int_{\mathcal{X}}f(d(\sigma^{-1}(p), p)) \phi(p) dp\nonumber                               \\
		               & \leq \mathcal{D}(Q)+  f(\zeta) \int_{\mathcal{X}} \phi(p) dp = \mathcal{D}(Q) + f(\zeta)\label{eq:cont-proof-1}.                                                        
	\end{align}
	Let $S^* = \{q_i^*| i \in [m]\}$ be the optimal configuration of the continuous problem, and $S^*_G$ be the configuration constructed by moving each robot location in $S^*$ to the closest sampled location in $\mathcal{S}$.
	Also note that,
	\begin{align}
		\mathcal{D}(S^*_G) & = \sum_{q_i \in S^*_G} \sum_{u \in W_i} c(u, q_i)w(u) = \sum_{i \in[m]} \sum_{u \in W_i} c(u, q_i)\int_{p \in \sigma(u)}\phi(p)dp\nonumber      \\
		                   & =\sum_{q_i \in S^*_G} \sum_{u \in W_i} \int_{p \in \sigma(u)}f(d(u, q_i))\phi(p)dp\nonumber                                                     \\
		                   & \leq \sum_{q_i \in S^*_G} \sum_{u \in W_i} \int_{p \in \sigma(u)}[\min_{q_j \in S^*_G}f(d(q_j, p)) + 2f(d(u, p)) + f(d(p,u))]\phi(p)dp\nonumber \\
		                   & \leq \mathcal{H}(S^*_G)  + 3f(\zeta)\label{eq:cont-proof-4},                                                                                    
	\end{align}
	where the first inequality is due to triangle inequality and Assumption~\ref{assump:metric-sensing}. Furthermore, we have,
	\begin{align}
		\mathcal{H}(S^*_G) & = \sum_{q_i \in S^*_G}\int_{V_i(S^*_G)} f(q_i, p) \phi(p)dp \leq \sum_{q_i \in S^*_G}\int_{V_i(S^*)} f(q_i, p) \phi(p)dp\nonumber \\
		                   & \leq \sum_{q_i \in S^*_G}\int_{V_i(S^*)} f(d(q_i^*, p)) \phi(p)dp + \int_{\mathcal{X}} f(d(q_i, q_i^*))\phi(p)dp \nonumber        \\
		                   & \leq \mathcal{H}(S^*) + f(\zeta) \label{eq:cont-proof-5},                                                                         
	\end{align}
	where the second inequality is due to triangle inequality and Assumption~\ref{assump:metric-sensing}.
	Let $Q^*_G$ be the optimal solution to the discrete coverage problem on graph $G$, then $D(Q) \leq \alpha \mathcal{D}(Q^*_G) \leq \alpha \mathcal{D}(S^*_G)$. Therefore, by Equations~\eqref{eq:cont-proof-1},~\eqref{eq:cont-proof-4} and~\eqref{eq:cont-proof-5}, we have, $\mathcal{H}(Q) \leq \alpha \mathcal{H}(S^*) + (4\alpha + 1)f(\zeta)$.\qed
\end{proof}

A $5$-approximation algorithm for the centralized coverage on metric graphs is provided in~\cite{arya2004local}. In the following section, we provide the first distributed approximation algorithm for the coverage in metric graphs.

\section{Distributed Algorithm On Graphs}
\label{sec:algorithm}
In distributed coverage control algorithms for continuous environments, and their adaptations to discrete environments, the algorithm drives each robot to the position inside its partition such that the sensing cost of its partition is minimized, i.e., the centroid of its Voronoi cell in the continuous problem. Although these algorithms converge to locally optimal solutions, there are no global guarantees on the quality of the solution. The following example provides a graph construction where such ``move to centroid'' algorithms perform poorly.

\begin{figure}
	\centering
	\begin{subfigure}{0.49\linewidth}
		\begin{center}
			\includegraphics[width=.7\linewidth]{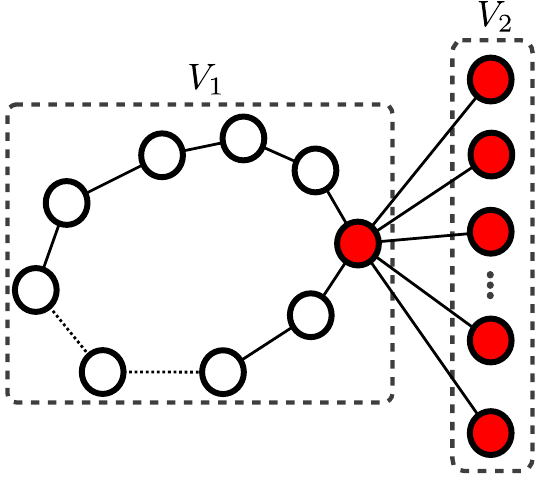}
			\caption{Locally optimal configuration under move to centroid control law}
			\label{fig:example-naive-algorithm-a}
		\end{center}
	\end{subfigure}
	\begin{subfigure}{0.49\linewidth}
		\begin{center}
			\includegraphics[width=.7\linewidth]{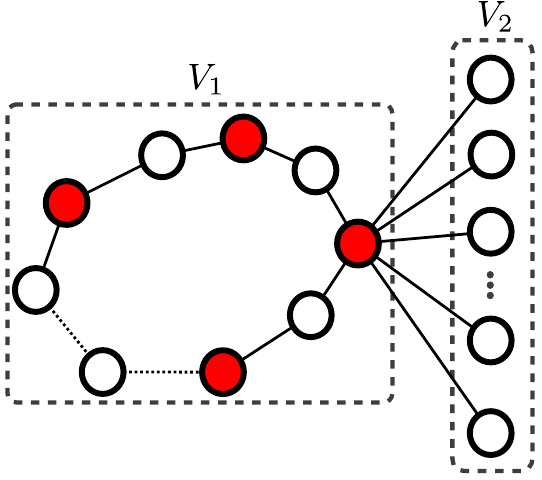}
			\caption{A better configuration}
			\label{fig:example-naive-algorithm-b}
		\end{center}
	\end{subfigure}
	\caption{Example environment with $3n + 1$ vertices and $n + 1$ robots}
	\label{fig:example-naive-algorithm}
\end{figure}

\begin{example}
	Consider the environment shown in Figure~\ref{fig:example-naive-algorithm} with $3n + 1$ vertices, $n + 1$ robots and unit costs for the shown edges. We consider the metric completion of the shown graph. The vertices are partitioned into two subsets: 1) $V_1$ with $2n + 1$ vertices and unit weights on the vertices and  2) $V_2$ with $n$ vertices of weights $\epsilon$ for some $0 < \epsilon \ll 1$. The highlighted vertices show the configuration of the robots. The configuration in Figure~\ref{fig:example-naive-algorithm-a} is a locally optimal solution under the move to centroid control law with global cost of $n(n+1)$. However, the configuration shown in Figure~\ref{fig:example-naive-algorithm-b} provides a global cost of $n + n\epsilon$. Therefore, the locally optimal solution provided by the move to centroid algorithm provides a solution with cost at least $\frac{n +1}{1 + \epsilon}$ of the optimal cost on the shown instance.
\end{example}

\subsection{High-level Idea}
\label{sec:high-level-idea}

The basic idea of our distributed coverage algorithm is to imitate the local-search algorithm for the $k$-median problem (See Section~\ref{subsec:discrete_env_problem_def}), namely \textsc{CentralizedAlg}, in a distributed manner.  The challenge in performing a local move in the distributed manner is that the robots are only aware of the partitions of their neighboring robots, therefore, the effect of a local move on the global objective is not known to the robots. However, we break down a local move in \textsc{CentralizedAlg} into a sequence of moves between neighbors. Let robot $j$ with position $q_j$ and neighbors $\mathcal{N}(j)$ be the closest robot to vertex $v$. Then a local move of \textsc{CentralizedAlg} swapping the position $q_i$ of robot $i$ with vertex $v$ is equivalent to a sequence of swaps inside $Q$ between the neighboring robots and a move from $q_j$ to $v$.  Figure~\ref{fig:k-median-move} shows an example of a local move in the centralized algorithm performed by a sequence of local moves. 

For this distributed coverage algorithm, we define the minimum communication range and neighbouring robots as follows:

\begin{figure}[t]
	\centering
	\includegraphics[width=.45\linewidth]{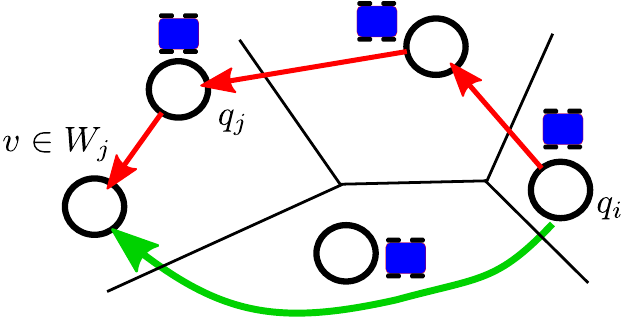}
	\caption{Local move in the centralized algorithm (green) and its equivalent sequence of moves in the distributed algorithm (red).}
	\label{fig:k-median-move}
\end{figure}

\begin{definition}[Neighbour robots]
	\label{assump:connection-assumption}
	Given a configuration $Q$, the set of neighbours of robot $i$ is defined as \[\mathcal{N}(i) = \{j \in [m]| d(q_i,q_j)\leq 4 \max \{\max_{p \in V_i} d(p, q_i), \max_{p \in V_j} d(p, q_j)\}\}\]
	where $V_i$ is the Voronoi partition of robot $i$ in the continuous environment.
\end{definition}

\begin{remark}
	The conventional definition of neighbours in the literature~\cite{cortes2004coverage} is that two robots are neighbours if the intersection of their Voronoi cell boundaries is not empty. Therefore, the distance of two neighbouring robots $i$ and $j$ can be  $2\max\{\max_{p \in V_i}d(q_i, p), \max_{p \in V_j}d(q_j, p)\}$.
	In~\cite{yun2014distributed}, authors show that in environments represented as graphs, with the conventional communication range, a move inside a robots partition might change  the partition of the neighbours of their neighbours. Therefore, they assume that the robots communicate with the neighbours of their neighbours, which is analogous to twice the communication range needed to implement the Lloyd's descent-based algorithms in continuous environments. In Section~\ref{sec:simulations}, we evaluate the performance of the algorithm with both the conventional and our definition of neighbours.
\end{remark}

We extend the definition of neighbouring robots in a continuous environment to capture the cases where the graph instance for the discrete coverage problem is given and the underlying continuous environment is unknown. Consider a graph instance $G = (V, E, c)$, then for each edge $(u, v) \in E$ we add a dummy vertex $z$ with zero weight and replace edge $(u, v)$ by two edges $(u, z)$ and $(z, v)$ such that $c(u, v) = 2c(u, z) = 2c(z, v)$. We let $c(u, v)$ for $(u, v) \notin E$ be the length of the shortest path in $G$ between $u, v$. Let $W_i$ be the partition of robot $i$, in the resulting graph. Then the equivalent definition of the neighbouring robots is given as follows:

\begin{definition}[Neighbour robots in graphs]
	\label{assump:connection-assumption}
	Given a configuration $Q$, the set of neighbours of robot $i$ is defined as \[\mathcal{N}(i) = \{j \in [m]| c(q_i,q_j)\leq 4 \max \{\max_{p \in W_i} c(p, q_i), \max_{p \in W_j} c(p, q_j)\}\}.\]
\end{definition}

\subsection{Detailed Description}
\label{sec:detailed-description-alg}

We are now ready to provide a detailed description of the proposed algorithm.

\emph{Algorithm Framework (Algorithm~\ref{alg:sender}):} For the ease of presentation, we provide a description of the algorithm in which the robots perform local moves sequentially. Each robot is assigned a unique identifier UID. Starting from an active robot, say robot $i$, the robot will make the possible local moves using Algorithm~\ref{alg:local-move}. If it can not make a local move, the robot will become inactive and will send a completion message to the neighbouring robots via \textsc{SendCompletionMessage} (line~\ref{algln:completionmessage} of Algorithm~\ref{alg:sender}). After execution of the local move by a robot, the robot becomes inactive. The next active robot to execute the local move can be selected in a distributed manner using a token passing algorithm~\cite{lemaire2004distributed}, or any other method that ensures each robot gets a turn at making a local move. The process terminates when all the robots become inactive. 

\begin{algorithm}[t]
	\begin{algorithmic}[1]
		\State{Each robot sets itself to active}
		\While{there exists an active robot}
		\For{any active robot $i\in \{1,\ldots,m\}$ }
		\If{$\sum_{u\in W_i} c(u, q_i) > 0$}
		\State{\textsc{LocalMove($i$)}}
		\State{\textsc{SendCompletionMessage}()}\label{algln:completionmessage}
		\EndIf
		\State {Robot $i$ deactivates}
		\EndFor
		\EndWhile
	\end{algorithmic}
	\caption{\textsc{DistributedCoverageAlgorithm}}
	\label{alg:sender}
\end{algorithm}

\emph{Local Move of Robot $i$ (Algorithm~\ref{alg:local-move}):} At an iteration of Algorithm~\ref{alg:local-move}, let the current configuration of robots be given by $Q = \{q_1,\ldots,q_m\}$ where the vertices in the partition of robot $i$ are given by $W_i$. The robot $i$ considers moving to a vertex $v\in W_i$ from its current vertex $q_i$. This move can only change the sensing cost of the vertices in the neighbouring robots' partitions (See Lemma~\ref{lem:connection-lemma} in Section~\ref{sec:analysis}). Hence, the robot $i$ can calculate the new neighboring partitions after a potential move to $v$. In line~\ref{algln:delta} of Algorithm~\ref{alg:local-move}, robot $i$ calculates the change in local objective $\delta_v$ due to this move for all $v\in W_i$. Since only robot $i$ is executing Algorithm~\ref{alg:local-move} at the current time, $\delta_v$ also represents the change in the global objective function. If $\min_{v\in W_i} \delta_v \leq -\epsilon_0$, robot $i$ moves to $q_i' = \argmin\delta_v$ and the iteration terminates
(local move type 1). If there is no valid local move of type 1, then robot $i$ calculates the change in local objective if it moves to $v$ and a new robot appears at $q_i$, i.e.,
\begin{equation}
	\label{eq:rho_v}
	\rho_v =  \sum_{j\in\mathcal{N}(i)}\sum_{u\in R_j(v)} w(u)[c(u,v) -  c(u,q_j)],
\end{equation}

where $R_j(v) = \{u \in W_j| c(u, q_j) > c(u, v)\}$ represents the vertices in the partition of robot $j$ that are closer to $v$ than the robot $j$ at $q_j$. Robot $i$ then passes the message with the set $\Gamma_i= \{\rho_v| v \in W_{i}\}$ and a counter set to $1$ to all its neighbors (line~\ref{algln:send-message} of Algorithm~\ref{alg:local-move}) and waits for a response (line~\ref{algln:receive-message} of Algorithm~\ref{alg:local-move}). If the response is a rejection from all the neighbors, Algorithm~\ref{alg:local-move} terminates. Otherwise it selects the acceptance message with the largest change in the objective and moves to the corresponding vertex. It also sends an acknowledgement message to the neighbor $k$ whose message was selected so that robot $k$ can move to $q_i$. 

\begin{algorithm}[t]
	\begin{algorithmic}[1]
		\While{$\exists$ a local move}
		\State Calculate $\delta_v$ for all vertices in $W_i$ \label{algln:delta}
		\If {$\min_{v} \delta_v \leq -\epsilon_0$}
		\State {Move to $v$}
		\Else
		\State Calculate $\Gamma_i= \{\rho_v| v \in W_{i}\}$ \Comment{Using Equation~\eqref{eq:rho_v}}
		\State \textsc{SendMessage}($\Gamma_i$, $1$) \label{algln:send-message}  
		\State \textsc{ReceiveMessage}() \label{algln:receive-message}
		\State \textsc{SendAcknowledgement}()
		\EndIf
		\EndWhile
	\end{algorithmic}
	\caption{\textsc{LocalMove($i$)}}
	\label{alg:local-move}
\end{algorithm}

\emph{Response of Other Robots (Algorithm~\ref{alg:receiver}):} When a robot $k$ receives messages from its neighbors, it follows Algorithm~\ref{alg:receiver}. Since messages from only one sender robot are propagating through the system at any time, it can select the message with the smallest counter value if it receives messages from multiple neighbors. The neighbor who sent the message with the smallest counter value is called the parent of robot $k$. It sends back a rejection message to all other neighbors. If the counter value of the message was one, it means that the message originated from its neighboring robot, say $i$. Then robot $k$ calculates the change in the sensing costs  of the vertices in $W_k\setminus R_k(v)$ if it moves to vertex $q_i$ and robot $i$ moves to vertex $v$ resulting in configuration $Q' = \{q_j | j \in \mathcal{N}(k)\}\cup\{v\}$, i.e., 
\begin{equation}
	\label{eq:ell_v}
	\ell_v = \sum_{u\in W_k\setminus R_k(v)}\min_{q\in Q'}[c(u,q) - c(u, q_k) ]w(u),
\end{equation}

If $\min_v (\rho_v+\ell_v) \leq -\epsilon_0$, robot $k$ decides to move to $q_i$ and sends an acceptance message to robot $i$ with the vertex $\argmin_v \rho_v+\ell_v$ and the change associated with this move. Otherwise it increments the counter and sends the message with $\Gamma_i$ to its neighbors. If the counter value in a message is greater than one, robot $k$ calculates $\ell$ as follows:
\begin{equation}
	\label{eq:ell}
	\ell = \sum_{u\in W_k}\min_{j\in\mathcal{N}(k)}[c(u,q_j) - c(u, q_k)]w(u)
\end{equation}

If $\min_v (\rho_v+\ell)\leq -\epsilon_0$, robot $k$ sends an acceptance message back to its parent. Otherwise it increments the counter and sends the message to its neighbors.

\begin{algorithm}[t]
	\hspace*{\algorithmicindent} \textbf{Input: } \texttt{message} = ($\Gamma_i$, \texttt{MessageCounter})
	\begin{algorithmic}[1]
		\If{\texttt{MessageCounter} = $1$}
		\State{Calculate $\ell_v$ for all $v$ in the message} \label{algln:line-l-v}\Comment{Equation~\eqref{eq:ell_v}}
		\If{$\exists v$ in the message with $\rho_v + \ell_v \leq -\epsilon_0$}
		\State{send acceptance message to parent}
		\Else
		\State{\textsc{SendMessage}($\Gamma_i$, \texttt{MessageCounter} + 1)}
		\EndIf
		\Else
		\State{Calculate $\ell$} \label{algln:line-l}\Comment{Equation~\eqref{eq:ell}}
		\If{$\exists v$ in the message with $\rho_v + \ell \leq -\epsilon_0$}
		\State{send acceptance message to parent}
		\Else
		\State{\textsc{SendMessage}($\Gamma_i$, \texttt{MessageCounter} + 1)}
		\EndIf
		\EndIf
		\State \textsc{ReceiveMessage}() \label{algln:rcvmsg}
	\end{algorithmic}
	\caption{\textsc{Receiver}}
	\label{alg:receiver}
\end{algorithm}

In function \textsc{ReceiveMessage} in line~\ref{algln:receive-message} of Algorithm~\ref{alg:local-move} and line~\ref{algln:rcvmsg} of Algorithm~\ref{alg:receiver}, if any robot receives at least one acceptance message from its neighbors, it passes the message with lowest increase in sensing cost value to its parent.  If it receives rejection messages from all its neighbors, it sends back a rejection message to its parent. Robot $i$ selects the move with maximum improvement in the sensing cost and sends back an acknowledgment using \textsc{SendAcknoledgment} to the accepted messages. Robots that receive the acknowledgment move to their parent's location. If there is no more local move available, robot $i$ sends a completion message using \textsc{SendCompletionMessage} to the neighbouring robots which will be propagated to all the robots. Then the next active robot executes \textsc{LocalMove}.

\section{Analysis of the Algorithm}
\label{sec:analysis}
In this section, we provide analysis on the quality of the solutions provided by the proposed algorithm. 
\subsection{Correctness and Approximation Factor}
Prior to providing the main results on the correctness and approximation factor of the algorithm, we provide two results on the change in the sensing cost of vertices in the partitions of the neighbouring robots with a move of a robot. 

The following result shows that a move by a robot inside its partition can only change the sensing cost of the vertices in its neighbouring partitions.
\begin{lemma}
	\label{lem:connection-lemma}
	Consider a vertex $z \in W_j$ where robot $j$ at position $q_j$ is the closest robot to $z$. Then robot $j$  is closer to $z$ than any vertex in the partition of a non-neighbour robot, i.e., $c(z, q_j) \leq \min_{i \notin \mathcal{N}(j)} \min_{u \in W_i} c(z, u)$.
\end{lemma}
\begin{proof}
    Proof by contradiction. Suppose there exists a move to vertex $v \in W_i$ by robot $i$ that changes the sensing cost of a vertex $z \in W_j$ of a non-neighbour robot $j$, i.e., $c(v, z) < c(q_j, z)$  which is equivalent to $d(v, z) \leq d(q_j, z)$ by the monotonicity of function $f$.
By the triangle inequality, we have,
\begin{equation}
\label{eq:connect-lem-proof-1}
d(q_j,v) \leq d(v, z) + d(q_j, z) \leq 2d(q_j, z). 
\end{equation}

Observe that $v \in W_i$, then $c(q_j , v) \leq c(q_i, v)$ which implies $d(q_j , v) \leq d(q_i, v)$. By Equation~\eqref{eq:connect-lem-proof-1}, we have $d(q_i,q_j) \leq d(q_j, v) + d(q_i, v) \leq 2d(q_j, v) \leq 4d(q_j, z)$. 

Then by the definition of the neighbouring robots in Section 4.1, robots $i$ and $j$ are neighbours. This is a contradiction.\qed

Observe that the proof of Lemma 1 holds for both definitions of neighbouring robots in continuous and discrete environments.

\end{proof}

Also, the following result shows that if a robot $i$ moves anywhere in the graph, then the vertices previously in $W_i$ will be assigned to robot in $\mathcal{N}(i)$.
\begin{lemma}
	\label{lem:connection-lemma-2}
	For any vertex $z \in W_i$, there exists a robot $j \in \mathcal{N}(i)$ at $q_j$ where 
	$c(z , q_j) \leq c(z, q_k)$ for all $ k \notin \mathcal{N}(i)$.
\end{lemma}
\begin{proof}

First we prove the result using definition of the neighbour robots in Definition 1. Suppose there exists $k \notin \mathcal{N}(i)$ and vertex $z \in W_i$ such that $c(z, q_k) < \min_{j \in \mathcal{N}(i)}c(z, q_j)$. Let $P$ be the shortest path from $z$ to $q_k$. Let $p$ be the point on the path where $P$ intersects the boundary of Voronoi cell of robot $i$. The point $p$ is not on the boundary of robot $k$, otherwise the distance between $d(q_i, q_k) \leq 4 \max\{\max_{p\in V_i} d(q_i, p), \max_{p\in V_k} d(q_k, p)\}$ which is a contradiction. Hence, the point $p$ is on the boundary of a neighbouring robot $j$. Therefore, there is a path from $q_j$ to $z$ shorter than $P$, then $c(q_j, z) = f(d(q_j, z)) \leq  f(d(q_k, z)) = c(q_k, z)$.

Now we prove the same result for the case where the underlying continuous coverage problem is unknown and neighbours are defined according to Definition 2. Suppose there exists $k \notin \mathcal{N}(i)$ and vertex  $z \in W_i$  such that $c(z, q_k) < \min_{j \in \mathcal{N}(i)}c(z, q_j)$. Observe that if two vertices of partitions $W_i$ and $W_k$ share an edge in $E$, then by the definition of neighbouring robots in Section A, robots $i$ and $k$ are neighbours. Therefore, since $k$ and $i$ are not neighbours, then there is no shared edge  between the vertices in $W_i$ and $W_k$. Let $P$ be the path on $G$ from $q_k$ to $z$. Then the path $P$ should contain a vertex $u$ in a partition of another robot $j$ which is a neighbour of robot $i$. Then by the metric property of the graph, $c(q_j, z) \leq c(q_j, u) + c(u, z) \leq c(q_k, u) + c(u, z) = c(q_k, z)$, where the second inequality is due to $u \in W_j$ and $c(u, q_j) \leq c(u, q_k)$. This is a contradiction. \qed

\end{proof}

Then we provide the following result on the change in the global objective with a successful move in the distributed algorithm.
\begin{lemma}
	\label{lem:same-cost-move}
	If a local move is accepted by the robot, then the global objective improves by at least $\epsilon_0$.
\end{lemma}
\begin{proof}
	A local move falls under the following cases:
	\begin{enumerate}
		\item Since, by Lemma~\ref{lem:connection-lemma}, a move of type 1 can only change the sensing cost of the neighbouring robots. Then the result is trivial for the local moves of type 1.
		      		      		      		      
		\item If the local move consists of a move by the robot $i$ that is executing \textsc{LocalMove} to a vertex $v$ in its partition and a neighbouring robot $j$ moving to vertex $q_i$. Let $Q'$ be the configuration after the local move, then the change in the global objective $\Delta \mathcal{D} = \mathcal{D}(Q') - \mathcal{D}(Q)$ is given by the following:
		      \begin{align}
		      	\Delta \mathcal{D} = \sum_{k \in [m]}\sum_{u \in  W_{k}} \min_{q \in Q'} w(u)c(u, q) - \sum_{k \in [m]}\sum_{u \in  W_{k}}  w(u)c(u, q_k) \label{eq:change-in-global-cost}. 
		      \end{align}
		      They by Lemma~\ref{lem:connection-lemma} and~\ref{lem:connection-lemma-2}, the sensing cost changes only for vertices in $u \in \cup_{k \in \mathcal{N}(i)} W_k$, therefore,
		      \begin{align*}
		      	\Delta \mathcal{D} & =  \sum_{k \in \mathcal{N}(i)}\sum_{u \in  W_{k}} w(u)[\min_{q \in Q'} c(u, q) - c(u, q_k)] \\
		      	                   & = \sum_{k \in \mathcal{N}(i)}\sum_{u \in  R_{k}(v)} w(u)[c(u, v) - c(u, q_k)]               \\&+ \sum_{k \in \mathcal{N}(i)}\sum_{u\in W_k\setminus R_k(v)}w(u)[\min_{q\in Q'}c(u,q)-c(u, q_k)].
		      \end{align*}
		      Observe that the sensing cost for vertex $u \in W_k\setminus R_k(v)$ for robot $k$ at $q_k \in Q' = \mathcal{N}(i)\cup\{v\}\setminus\{j\}$ does not change, therefore, we have
		      \begin{align*}
		      	\Delta \mathcal{D} = \sum_{k \in \mathcal{N}(i)} & \sum_{u \in  R_{k}(v)} w(u)[c(u, v) - c(u, q_k)] \\&+ \sum_{u\in W_j\setminus R_j(v)}w(u)[\min_{q\in Q'}c(u,q)-c(u, q_j)].
		      \end{align*}
		      		      		      		      
		      Hence, the result follows immediately as $\Delta \mathcal{D} = \rho_v + l_v$.
		      		      		      		      
		\item  Suppose the local move consists of a move by the robot $i$ that is executing \textsc{LocalMove} to a vertex $v$ in its partition and a sequence of moves between the neighbouring robots. Without loss of generality, let $\langle v, q_{i}, q_{i+1}, \ldots, q_{j - 1}, q_{j}\rangle$ be the sequence of moves between the neighbouring robots where each robot moves to the previous vertex of the preceding robot in the sequence. Let $Q'$ be the configuration after the local move, then the change in the global objective is given by Equation~\eqref{eq:change-in-global-cost}. Observe that each robot accepts only the message from the parent robot. Therefore, among the neighbours of the robots in the sequence only the parent of each robot moves.
		
		 First we show that the change in the sensing cost under this sequence of moves only occurs for vertices in $\cup_{k \in \mathcal{N}(i)} W_k$ and vertices in $W_j$. Let $u$ be a vertex assigned to robot $k' \in [m] \setminus \{j \cup \mathcal{N}(i) \}$ in configuration $Q$, i.e., $ u \in W_{k'}$. Since $k' \notin \mathcal{N}(i)$, then by Lemma~\ref{lem:connection-lemma} a move to vertex $v$ will not improve the sensing cost of $u$. Also if $k'$ is among the robots moving in the sequence, then there is a robot moving to its previous location, therefore, each vertex in $W_{k'}$ will be sensed by another robot with the same sensing cost. Therefore, the total change $\Delta \mathcal{D}$ in the sensing cost of the vertices becomes
		      \begin{align*}
		      	\sum_{k \in \mathcal{N}(i)} & \sum_{u \in  R_{k}(v)} w(u)[c(u, v) - c(u, q_k)] + \sum_{u\in W_j}w(u)[\min_{q\in Q'}c(u,q)-c(u, q_j)]. 
		      \end{align*}
		      		      		      		      
		      Hence, the result follows immediately as $\Delta \mathcal{D} = \rho_v + l$.\qed
	\end{enumerate}
\end{proof}

Now we show the following result on the valid local moves in \textsc{CentralizedAlg} given the final configuration of the proposed distributed algorithm.
\begin{lemma}
	\label{lem:no-move-remains}
	If the proposed distributed algorithm terminates, then there is no  single swap move in the centralized local search algorithm $\textsc{CentralizedAlg}$ (See Section~\ref{subsec:discrete_env_problem_def}) that improves the objective function.
\end{lemma}
\begin{proof}
	Suppose that there exists a centralized local move of robot $j$  at vertex $q_j$ to a vertex $v \in W_i$ that improves the objective function by $\epsilon_0$. Therefore, adding a robot to $v$ improves the sensing cost of the vertices in $\cup_{k\in \mathcal{N}(i)}W_k$ by $\rho_v\leq -\epsilon_0$. Therefore, by construction of the algorithm, the robot $i$ would have suggested the move to its neighbouring robots. Suppose after $l$ communications, robot $j$ at $q_j$ receives the message
	for the first time. In Line~\ref{algln:line-l-v} (resp. Line~\ref{algln:line-l})  of Algorithm~\ref{alg:receiver} if $j \in \mathcal{N}(i)$ (resp. $j \notin \mathcal{N}(i)$), robot $j$ calculates the increase in the sensing cost $\ell_v$ (resp. $\ell$) for the vertices in $\cup_{k\in \mathcal{N}(j)}W_k$ by the move from $q_j$ to the parent of robot $j$. Since robot $j$ has rejected this offer, by Lemma~\ref{lem:same-cost-move} the change in the global sensing cost is less than $\epsilon_0$. This is a contradiction.\qed
\end{proof}

\begin{theorem}
	\label{thm:disc-approx}
	The proposed distributed coverage control algorithm provides a solution within $5 + o(\epsilon_0)$ factor of the optimal configuration.
\end{theorem}
\begin{proof}
	The result follows immediately from Lemma~\ref{lem:no-move-remains}. The final configuration in the distributed algorithm is a locally optimal solution for the \textsc{CentralizedAlg} with single swap at each iteration, i.e. $p =1$, therefore, the configuration provides a coverage within  a factor $5 + o(\epsilon_0)$ of the global optimal.\qed
\end{proof}
\begin{corollary}
	Given an environment $\mathcal{X}$ with $m$ mobile robots and a sampling of $\mathcal{X}$ with dispersion $\zeta$, the solution $Q$ obtained from the proposed distributed coverage control algorithm provides coverage cost $\mathcal{H}(Q) \leq 5 \mathcal{H}(S^*) + o(f(\zeta) + \epsilon_0)$, where $S^*$ is the optimal solution of the continuous coverage problem.
\end{corollary}
\begin{proof}
	Proof follows immediately from Theorems~\ref{thm:proof-cont-approx} and~\ref{thm:disc-approx}.\qed
\end{proof}

\begin{remark}[Asynchronous Execution of Local Moves]
In the asynchronous implementation of the algorithm, instead of the robots performing local moves in turn, the robots calculate the change in local objectives and send and receive messages in parallel. If a robot receives messages originating from multiple active robots, it selects the message from one of them (for instance, from the robot with lowest UID) and runs Algorithm~\ref{alg:receiver} for that message. When a robot decides to move to a new location, it moves only if none of its neighbors are currently moving. If any of its neighbors is currently moving, it waits until all of its neighbors stop moving and runs Algorithm~\ref{alg:local-move} and Algorithm~\ref{alg:receiver} again.  

As shown in the Lemmas~\ref{lem:connection-lemma} and~\ref{lem:connection-lemma-2}, if the robot $i$ moves in its partitions or to a vertex $v$ in the partition of robot $j$, only the cost of the vertices in $\mathcal{N}(i)$ and $\mathcal{N}(j)$ is affected. Therefore, if a robot only moves when none of its neighbors are moving, the change in local cost calculated by that robot is correct. Hence the distributed algorithm presented in the paper can be implemented in an asynchronous fashion.
\end{remark}

\section{Time Complexity}
\label{sec:time-complexity}
In this section, we characterize the runtime and the communication complexity of the proposed algorithm.

Let $Q_0$ be the starting configuration of the robots and $Q^*$ be the optimal configuration for problem of coverage on graph $G$, then we have the following result on the runtime of the proposed algorithm. For $\epsilon_0 = 0$, we follow the the analysis similar to~\cite{yun2014distributed}. Since there are a finite number of possible local moves that improve the global sensing cost, each iteration improves the global sensing cost by at least $\epsilon'>0$. Therefore the algorithm terminates in $\frac{\mathcal{D}(Q_0) - \mathcal{D}(Q^*)}{\epsilon'}$ iterations. Observe that with $\epsilon_0 = 0$, the distributed algorithm provides a solution within $5$ factor of the optimal solution with possibly non-polynomial number of iterations. However, we can prove convergence in polynomial time if the sampling of the environment $\mathcal{S}$ satisfies the following properties. 
\begin{assumption}
The weight of a vertex $v$ in $\mathcal{S}$ is at least $w_0$ for some $w_0 > 0$, i.e., $\int_{p \in \sigma(v)} \phi(p)dp \geq w_0 >0$.
\end{assumption}

This follows the common assumption in the coverage control literature where there is a basis function defined for $\phi$~\cite{cortes2004coverage}. For a given $w_0$, we remove a vertex $v$ with with weight $w(v) < w_0$ from the samples and recalculate the weights on the vertices. 

\begin{assumption}
The ratio $\max_{u, v \in \mathcal{S}} f(d(u, v))/\min_{u, v \in \mathcal{S}} f(d(u, v))$ is polynomial in the number of the samples $|\mathcal{S}|$. 
\end{assumption}

For instance, if a graph is constructed via a grid sampling of the continuous environment, where each cell is a $d\times d$ square, then we have
\[
    \frac{\max_{e \in E} c(e)}{\min_{e \in E} c(e)}  = \frac{f(\sqrt{2|\mathcal{S}|}d)}{f(d)} =  O(\sqrt{|\mathcal{S}|}),
\]
where the second equality is by the sub-additivity of sensing function $f$.

For a given $\epsilon >0$ and a polynomial $p(|\mathcal{S}|, m)$, we have,
\begin{lemma}
The proposed algorithm with $\epsilon_0 = \frac{\epsilon w_0 }{p(|\mathcal{S}|, m)}\min_{e \in E}c(e)$ terminates in the polynomial number of iterations, i.e.,
\[\log(\mathcal{D}(Q_0)/\mathcal{D}(Q^*))/\log(1 - \frac{\epsilon w_0 }{p(|\mathcal{S}|, m)} \frac{\min_{e \in E}c(e)}{\max_{e \in E}c(e)}).\]
\end{lemma}

\begin{proof}
Let $Q_i$ be the configuration of the robots at step $i$ of the algorithm. As the global sensing cost improves by at least $\epsilon_0$ after each iteration, we have,
\[
\mathcal{D}(Q_{i + 1}) - \mathcal{D}(Q_{i}) \leq - \epsilon_0 = - \frac{\epsilon w_0 }{p(|\mathcal{S}|, m)}\min_{e \in E}c(e).
\]
Observe that $\mathcal{D}(Q_{i}) \leq  \max_{e \in E}c(e)$. Therefore, 
\[
\mathcal{D}(Q_{i + 1}) \leq \mathcal{D}(Q_{i})- \frac{\epsilon w_0 }{p(|\mathcal{S}|, m)}\min_{e \in E}c(e) \leq \mathcal{D}(Q_{i})(1 - \frac{\epsilon w_0}{p(|\mathcal{S}|, m)} \frac{\min_{e \in E}c(e)}{\max_{e \in E}c(e)}).
\]
Therefore, at each iteration of the distributed coverage algorithm the sensing cost improves by the factor of $1 - \frac{\epsilon w_0}{p(|\mathcal{S}|, m)} \frac{\min_{e \in E}c(e)}{\max_{e \in E}c(e)}$. Hence, the algorithm terminates in $\log(\mathcal{D}(Q_0)/\mathcal{D}(Q^*))/\log(1 - \frac{\epsilon w_0 }{p(|\mathcal{S}|, m)} \frac{\min_{e \in E}c(e)}{\max_{e \in E}c(e)})$ iterations which is polynomial in the input size, $1/\epsilon$ and $1/w_0$.\qed
\end{proof}

\begin{remark}
At each iteration of the algorithm, at most $m^2$ messages are sent with size at most $n\log(\max_{e\in E}c(e)) +  \log(m)$ bits by the robot executing \textsc{LocalMove}. Then at most $m^2$ messages are sent back between the robots in the acceptance/rejection step. Finally, the \textsc{SendAcknowledgment} step requires at most $m$ messages. Hence, the communication complexity of the proposed distributed algorithm is $O(\log(\mathcal{D}(Q_0)/\mathcal{D}(Q^*))/\log(1 - \frac{\epsilon w_0 }{p(|\mathcal{S}|, m)} \frac{\min_{e \in E}c(e)}{\max_{e \in E}c(e)})m^2)$.
\end{remark}

\section{Simulation Results}
\label{sec:simulations}
In this section, we evaluate the performance of the proposed distributed algorithm and compare it to convex and non-convex distributed coverage algorithms and the  centralized algorithm in~\cite{ahmadian2013local}. To construct the discrete problem described in Section 3, we use a grid sampling of the environment. We denote the maximum distance inside a Voronoi cell of a robot $i$ by $R_{\mathrm{comm}} = \max_{p \in V_i} d(q_i, p)$ and we evaluate the performance of the proposed algorithm with communication ranges of $4R_{\mathrm{comm}}$ (See Definition 1) and $2R_{\mathrm{comm}}$ which is analogous to the conventional communication model in the continuous coverage literature. In the rest of this section, we use $f(x)=x$ as the sensing cost function.

\subsection{Convex Environments}
In this experiment, we compare our algorithm to distributed Lloyd's algorithm~\cite{cortes2004coverage} in convex environments. We use the Euclidean distance as the metric between two points. The comparison is conducted in a $1500\times 850$ environment with $100$ different event distributions. The event distributions are truncated multivariate normal distributions with mean $[1400, 800]$ and  covariance matrices $\Sigma = [\sigma, 0; 0,\sigma]$ where $\sigma$ is uniformly randomly selected from interval $ [5,10] \times 10^4$. In this experiment, the robots are initialized in the bottom left corner of the environment. Figure~\ref{fig:convex_diff_robots} illustrates the percentage difference of the solutions provided by the two algorithms with respect to the solution of the centralized algorithm. Observe that the proposed achieves solution quality very close to the centralized algorithm, even with a large number of robots, while Lloyd's algorithm provides solutions with approximately $15\%$ deviation.  

\begin{figure}
    \centering
    \includegraphics[width=.8\textwidth]{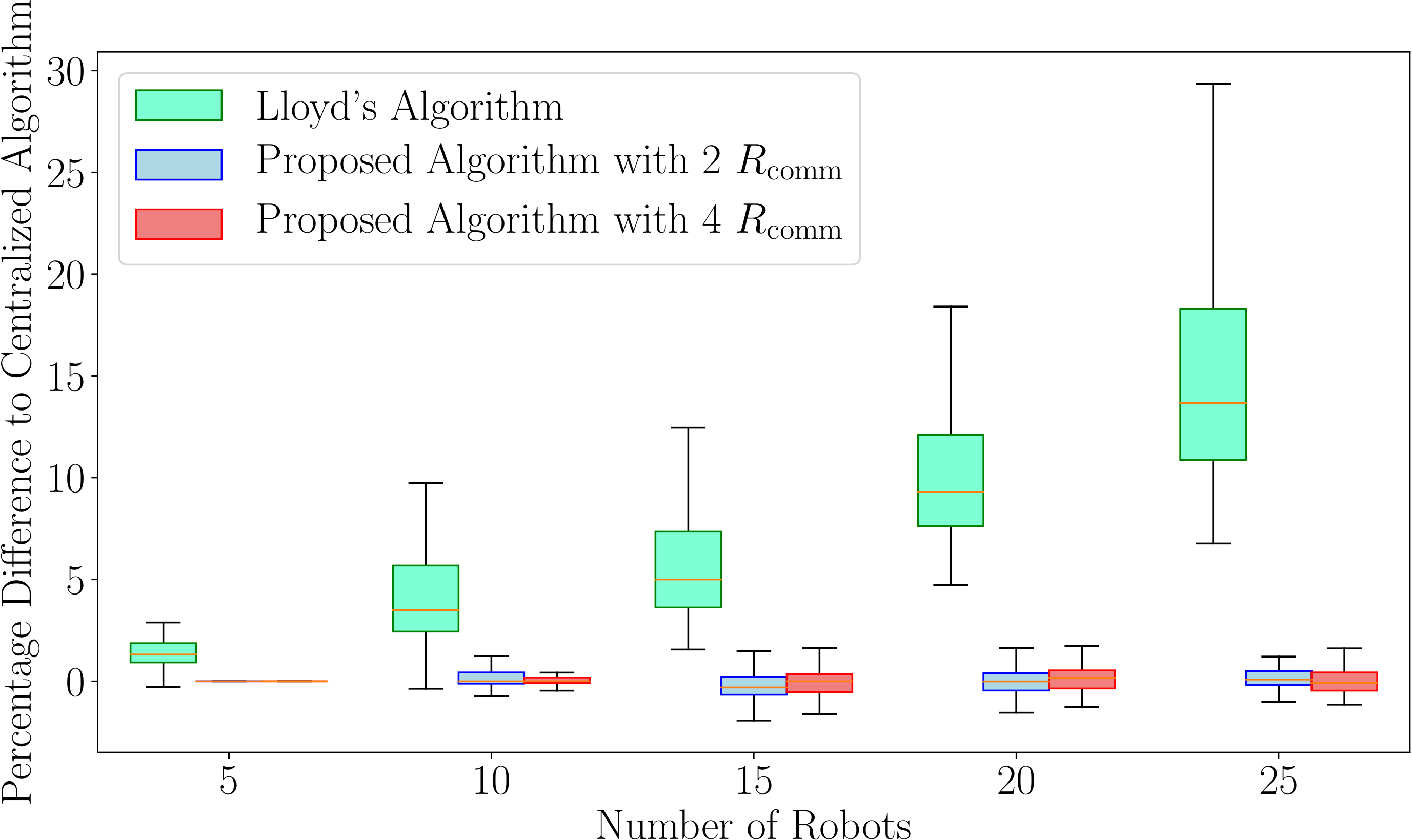}
    \caption{Percentage difference of the solutions of different algorithms to the solution of the centralized algorithm}
    \label{fig:convex_diff_robots}
\end{figure}

\begin{figure}[H]
    \centering
    \includegraphics[width=.75\textwidth]{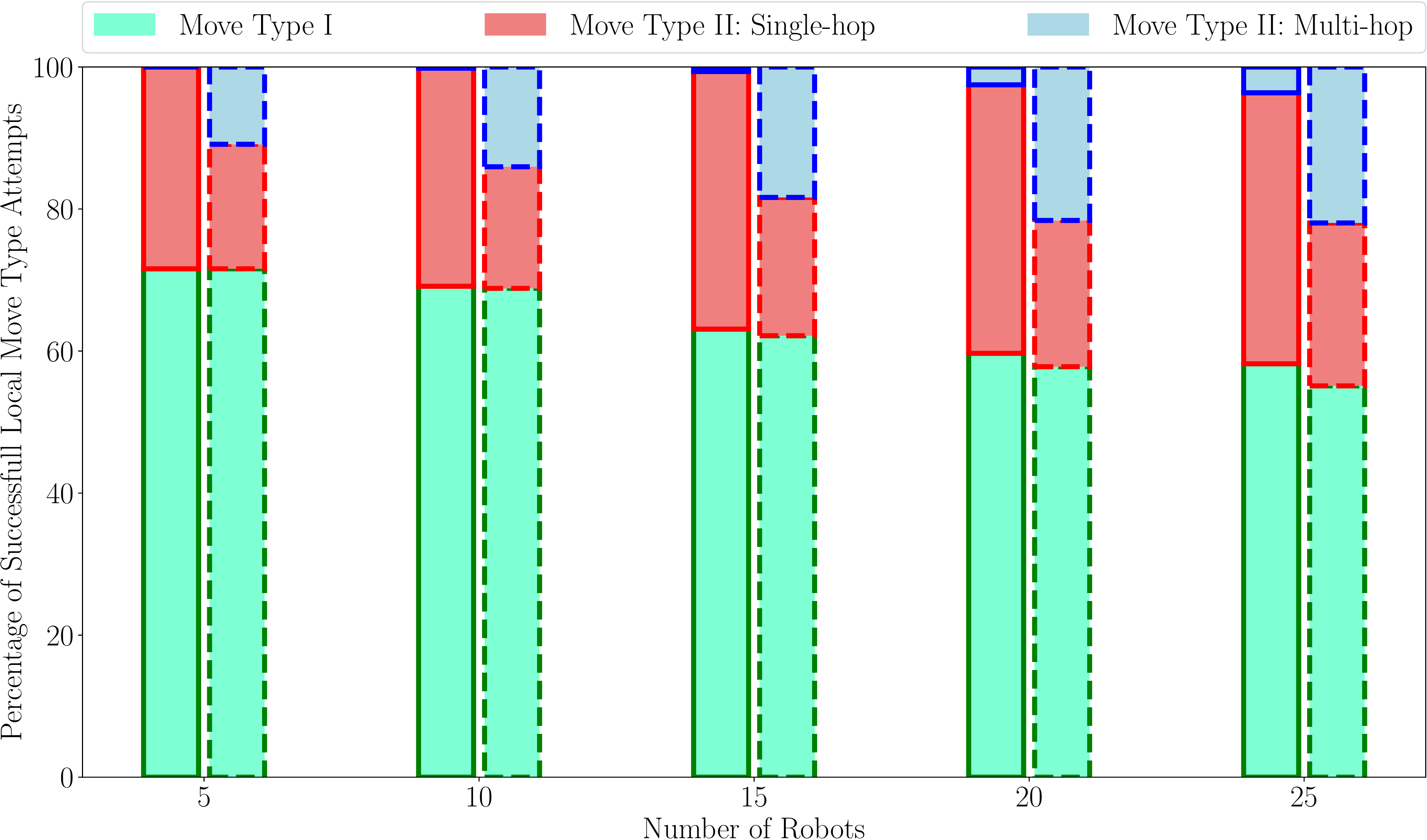}
    \caption{Percentage share of attempts for different local move types which resulted in improving the objective.}
    \label{fig:convex_diff_robots_msg}
\end{figure}

Figure~\ref{fig:convex_diff_robots_msg} shows the percentage share of the different local move types. The dashed lines show the results  with $2R_{\mathrm{comm}}$ communication range.  The single hop Move type II is a local move which involves a robot and its neighbours in comparison to the multi-hop local move which involves non-neighbour robots. Observe that the majority of the local moves are Move Type I which communicates with only the neighbouring robots. However, the local Moves of Type II help the proposed algorithm to leave locally optimal solutions.

\begin{figure}
    \centering
    \includegraphics[width=.8\textwidth]{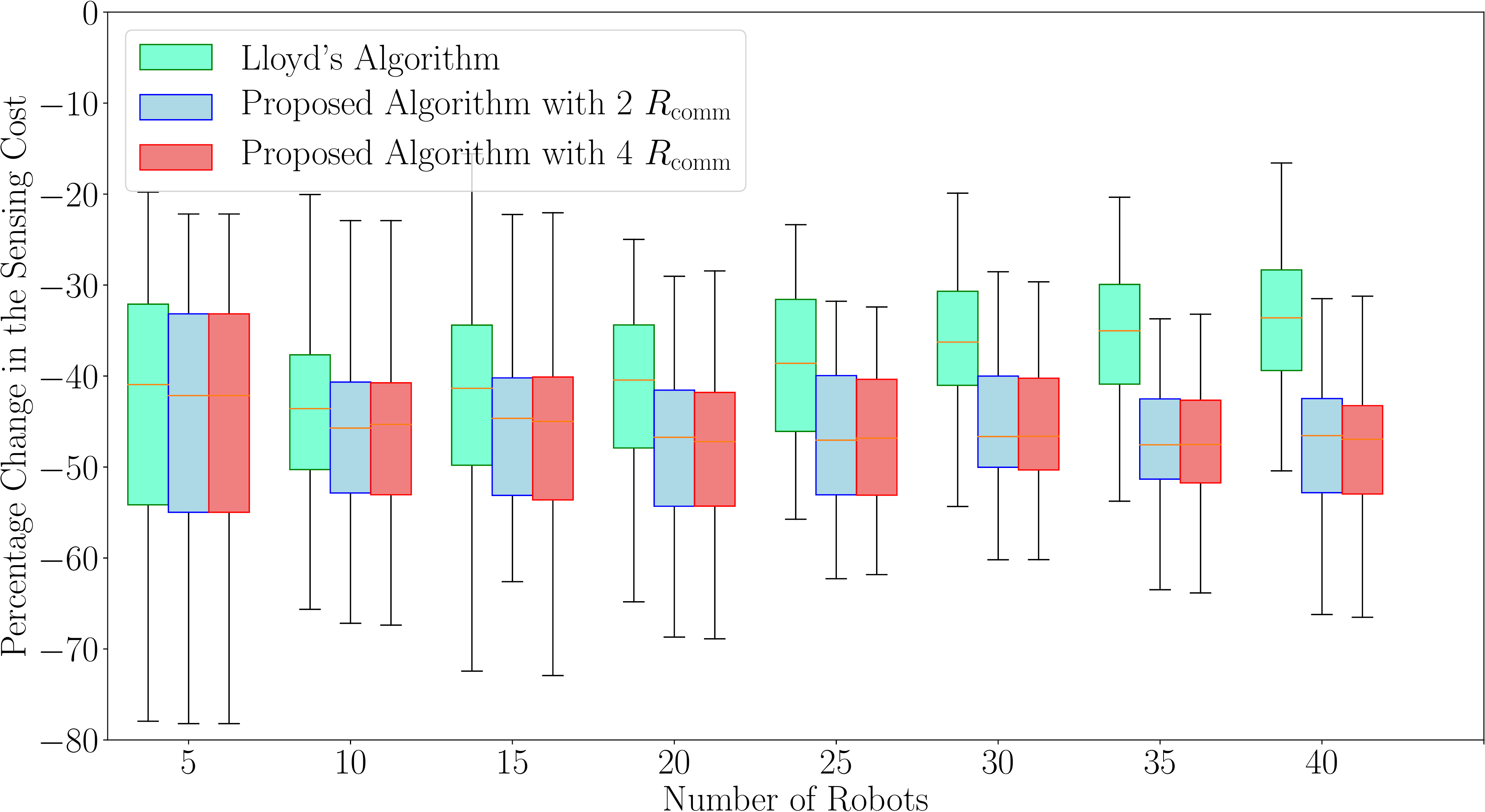}
    \caption{Percentage improvement of algorithms with random initial configurations.}
    \label{fig:convex_random_init_configs}
\end{figure}

Figure~\ref{fig:convex_random_init_configs} illustrates the improvement in the sensing cost for the Lloyd's and the proposed algorithm in a convex environment with uniformly random initial configurations. The results are obtained for $50$ initial configurations in the environment. Observe that even with the large number of robots where the random configuration provides a relatively good sensing cost, the proposed algorithm improves the solution by $50\%$ on average. In a system of $40$ robots, our proposed algorithm provided $\approx 15\%$ additional improvement on the sensing cost as compared to the Llyod's algorithm.

Figure~\ref{fig:final_config_convex} shows the final configuration and the paths of $10$ robots for the two algorithms in a test environment.

\begin{figure}[H]
    \begin{subfigure}[t]{0.49\linewidth}
    \begin{center}
    \includegraphics[width=\linewidth]{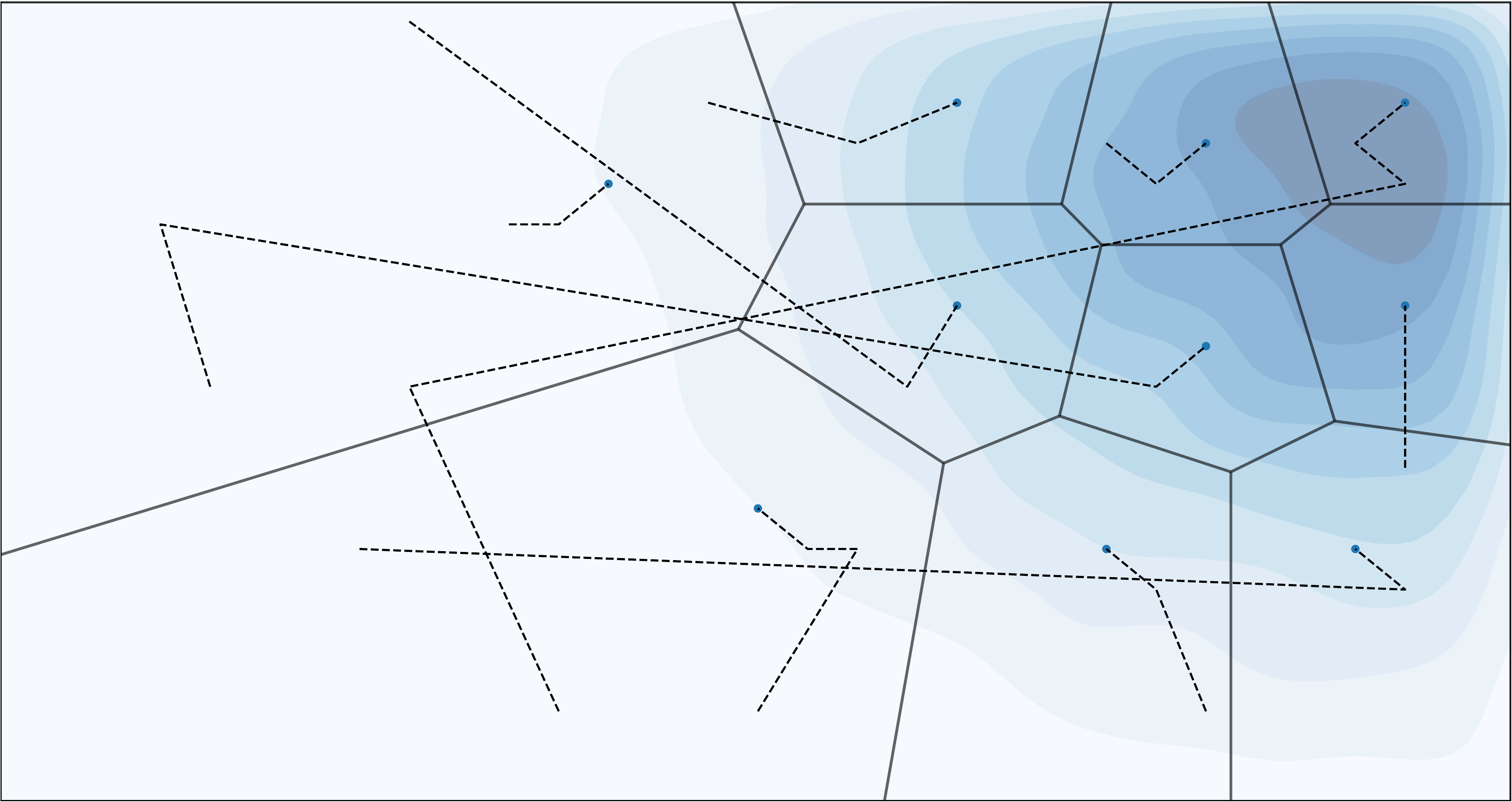}
    \caption{Proposed Algorithm with $4 R_{\mathrm{comm}}$ communication range.}
    \end{center}
    \label{fig:ours-path-a}
    \end{subfigure}
    \begin{subfigure}[t]{0.49\linewidth}
    \begin{center}
    \includegraphics[width=\linewidth]{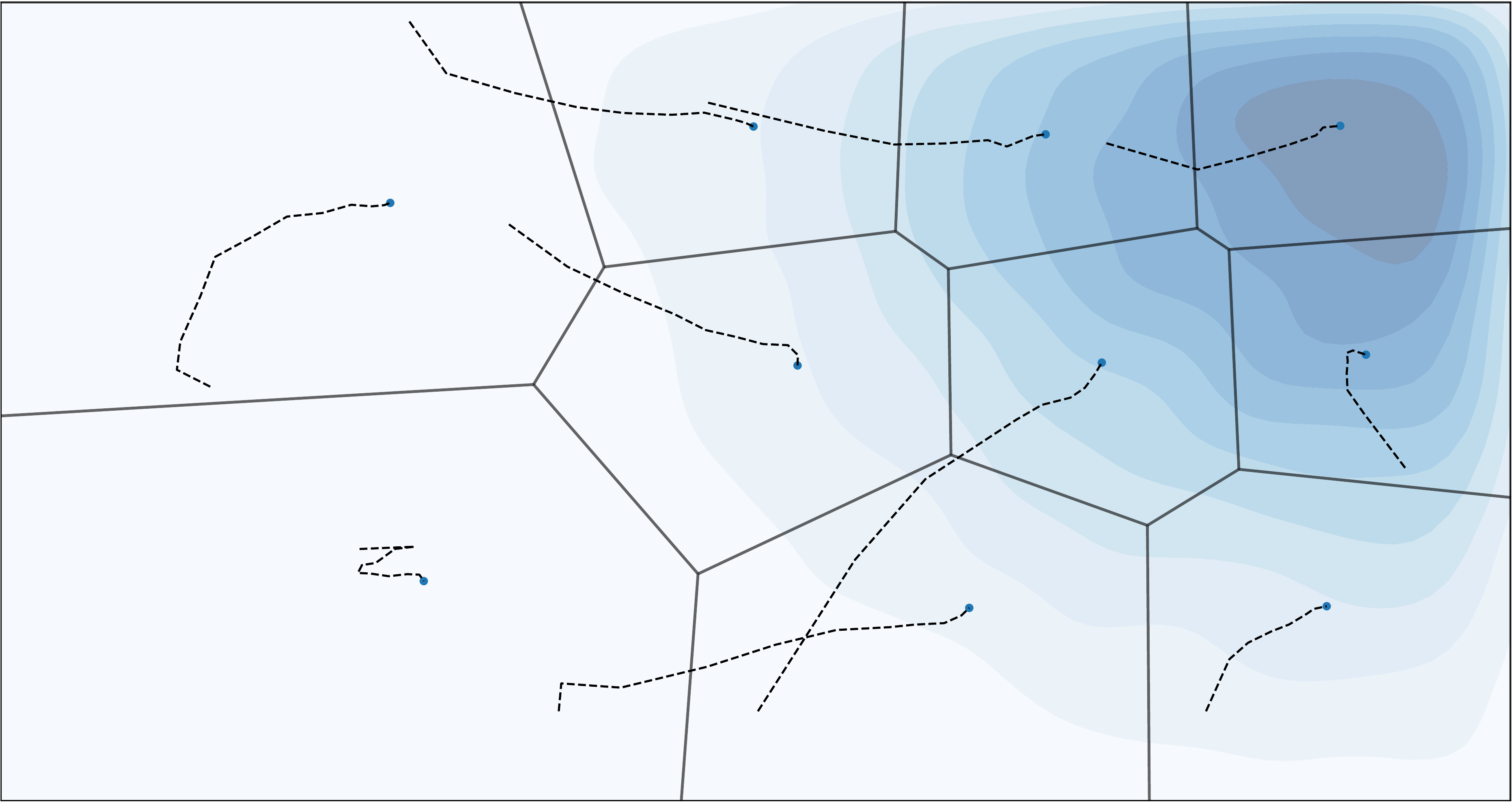}
    \caption{Lloyd's Algorithm}
    \end{center}
    \label{fig:lloyds-path-b}
    \end{subfigure}
    \caption{Final configuration of the robots for the two algorithms in a test environment}
    \label{fig:final_config_convex}
\end{figure}

\subsection{Non-Convex Environments}
In this section, we compare the solution quality of the proposed algorithm with two different communication ranges to the algorithms in~\cite{breitenmoser2010voronoi},~\cite{yun2014distributed} and the centralized algorithm~\cite{ahmadian2013local}. The experiment is conducted in a  $1500\times 850$ environment that contains obstacles (See Figure~\ref{fig:sample_Env_extended_sim-a}), and using $100$ different event distributions. The distributions are generated in the same manner as in the convex environment experiments with uniformly random mean and covariance matrices.
The communication model in the non-convex studies are different, for instance, two robots are neighbours in~\cite{breitenmoser2010voronoi}, if the intersection of the Voronoi cells of the robots in the environment without obstacles is non-empty, and two robots are neighbours in~\cite{yun2014distributed} if the two partitions of the robots share an edge in the discrete representation of the environment. Therefore in the implementations of algorithms in~\cite{breitenmoser2010voronoi} and \cite{yun2014distributed}, we assume that robots are connected to every other robot. 
\begin{figure}[t]
    \begin{subfigure}[t]{0.49\linewidth}
    \begin{center}
    \includegraphics[angle=90, width=.7\linewidth]{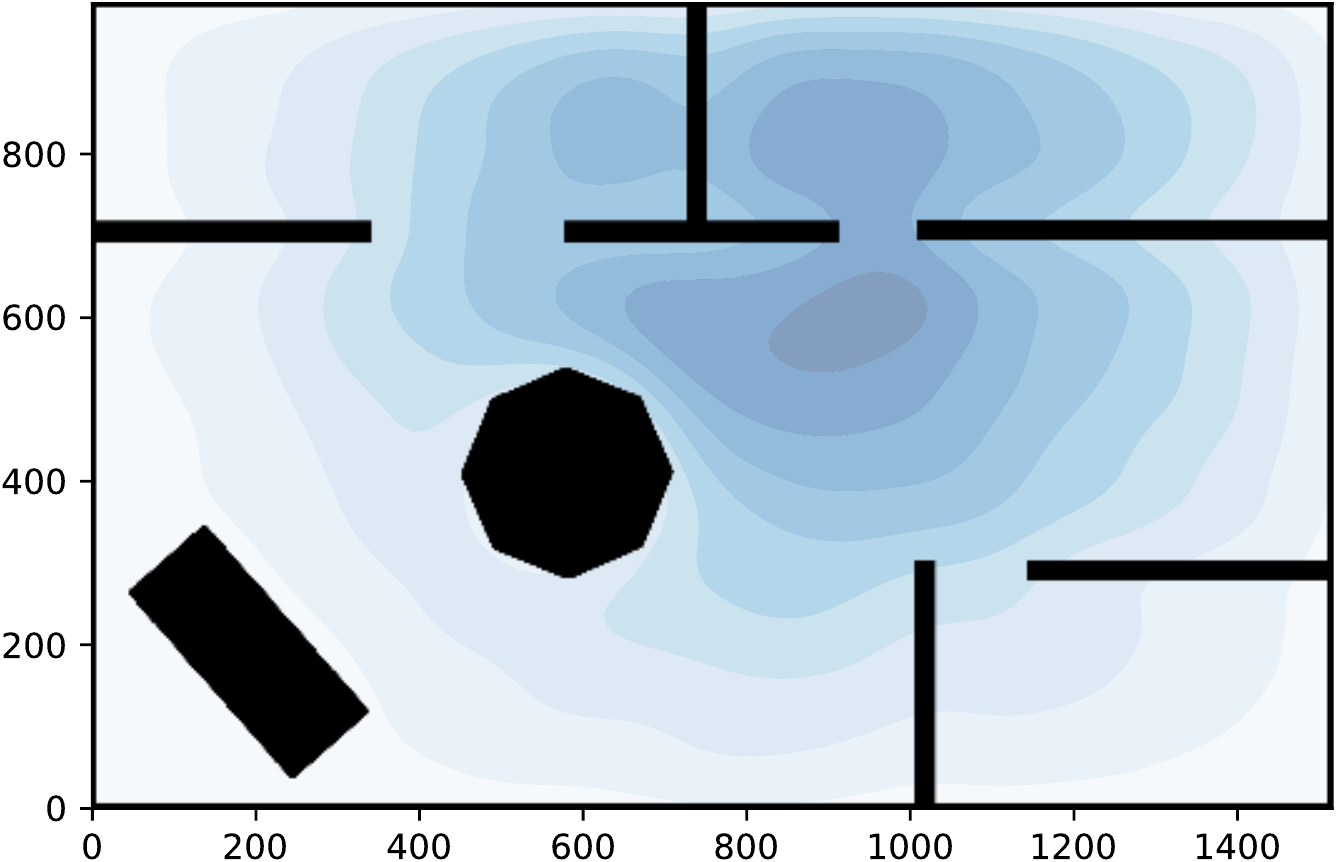}
    \end{center}
    \caption{Sample Environment with Obstacles}
    \label{fig:sample_Env_extended_sim-a}
    \end{subfigure}
    \begin{subfigure}[t]{0.49\linewidth}
    \begin{center}
    \includegraphics[angle=90, width=.7\linewidth]{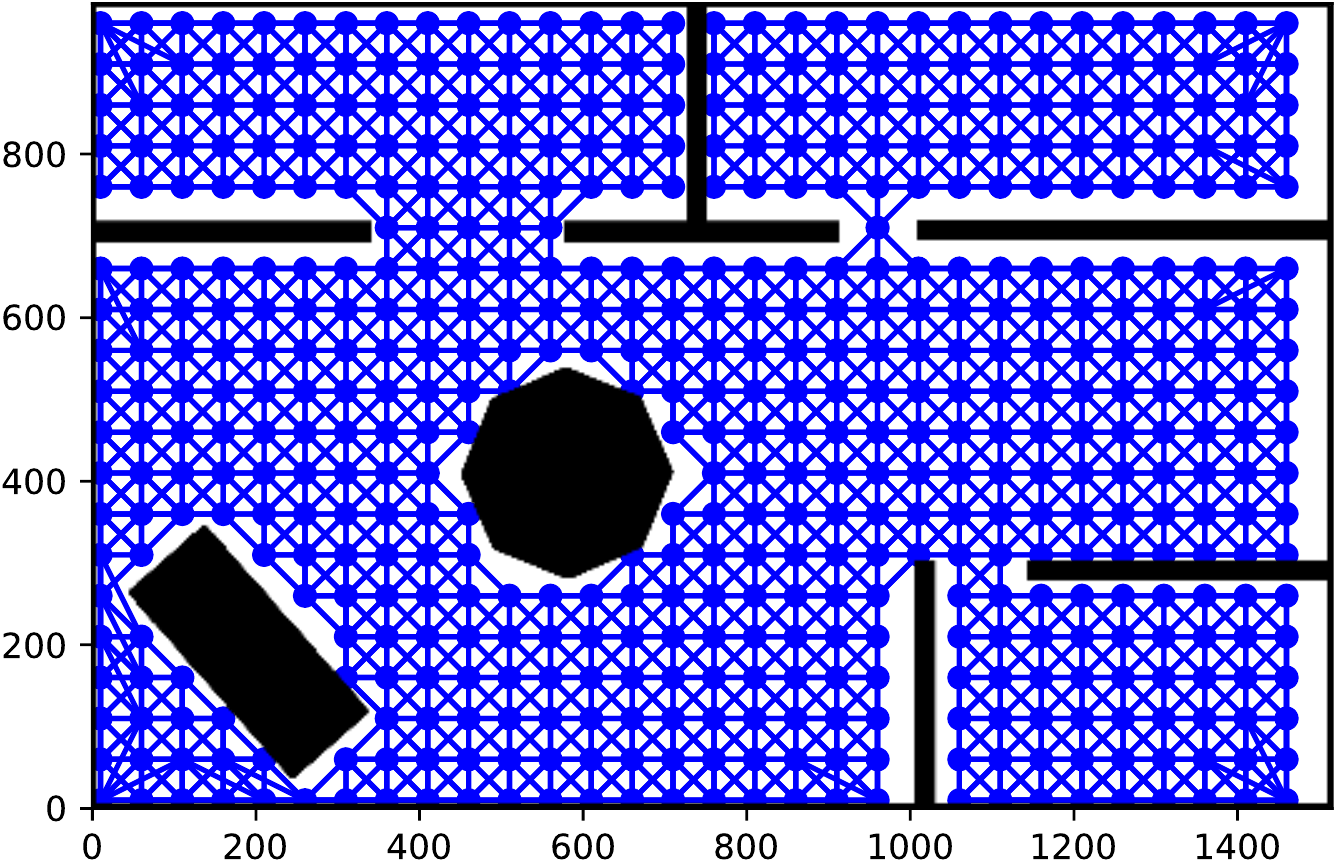}
    \end{center}
    \caption{Discrete representation of the environment}
    \label{fig:sample_Env_graph_extended_sim-b}
    \end{subfigure}
    \caption{A sample environment with obstacles and its discrete representation}
    \label{fig:aaaaa-a}
\end{figure}

 Figure~\ref{fig:non_convex_diff_robots} shows the percentage difference between the solutions of each algorithm compared to the centralized algorithm. Observe that the proposed algorithm even with the conventional communication range out-performs both other algorithms by $\approx 20\%$ on average in a system with $30$ robots and matches the solution quality of the centralized algorithm. Figure~\ref{fig:non_convex_paths} illustrates the final configuration and the movement of the robots using the proposed algorithm in a non-convex environment.
\begin{figure}
    \centering
    \includegraphics[width=.95\textwidth]{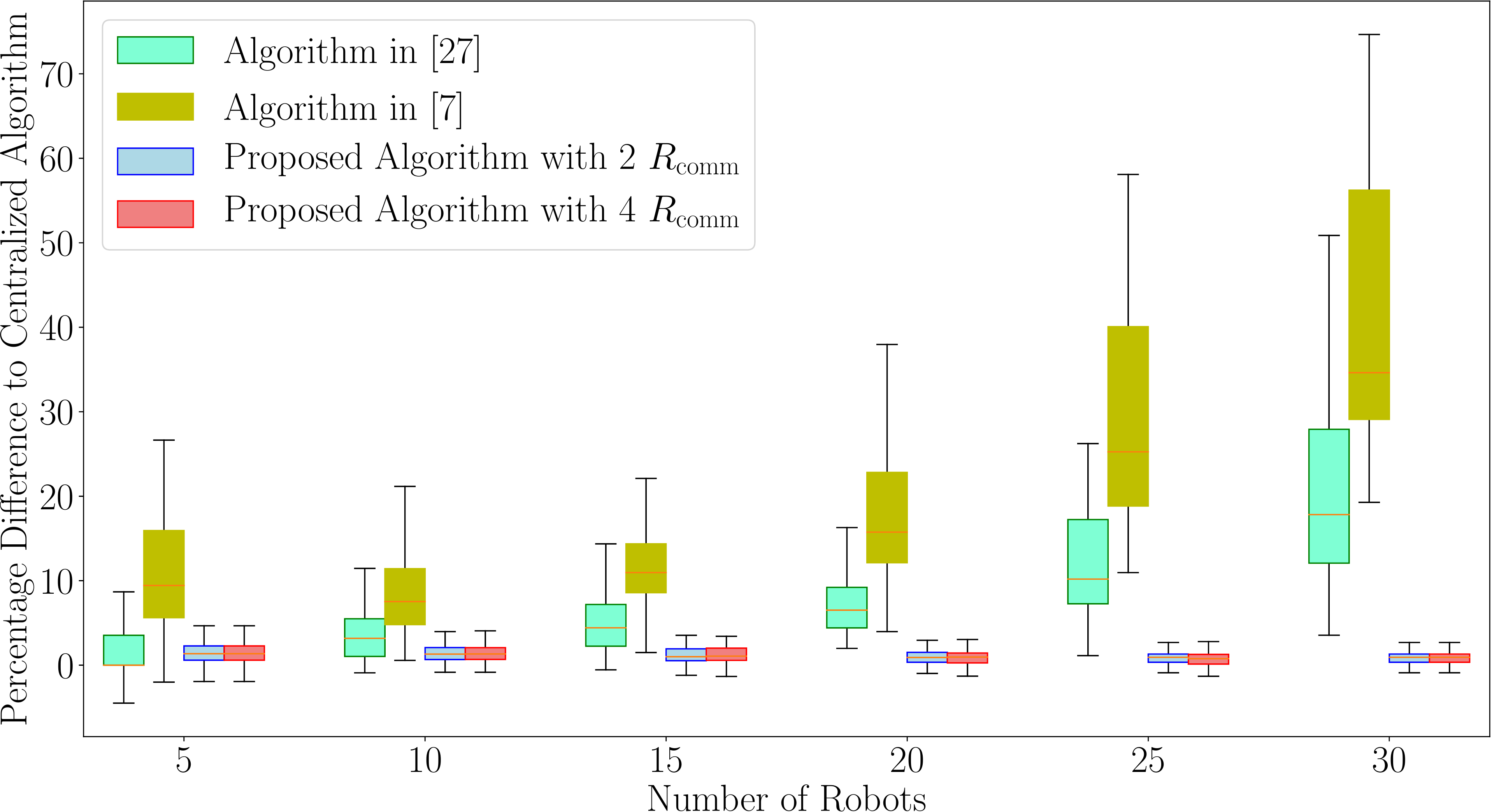}
    \caption{Percentage difference of the solutions of different algorithms to the solution of the centralized algorithm}
    \label{fig:non_convex_diff_robots}
\end{figure}

\begin{figure}
\centering
    \includegraphics[angle=-90,width=.7\textwidth]{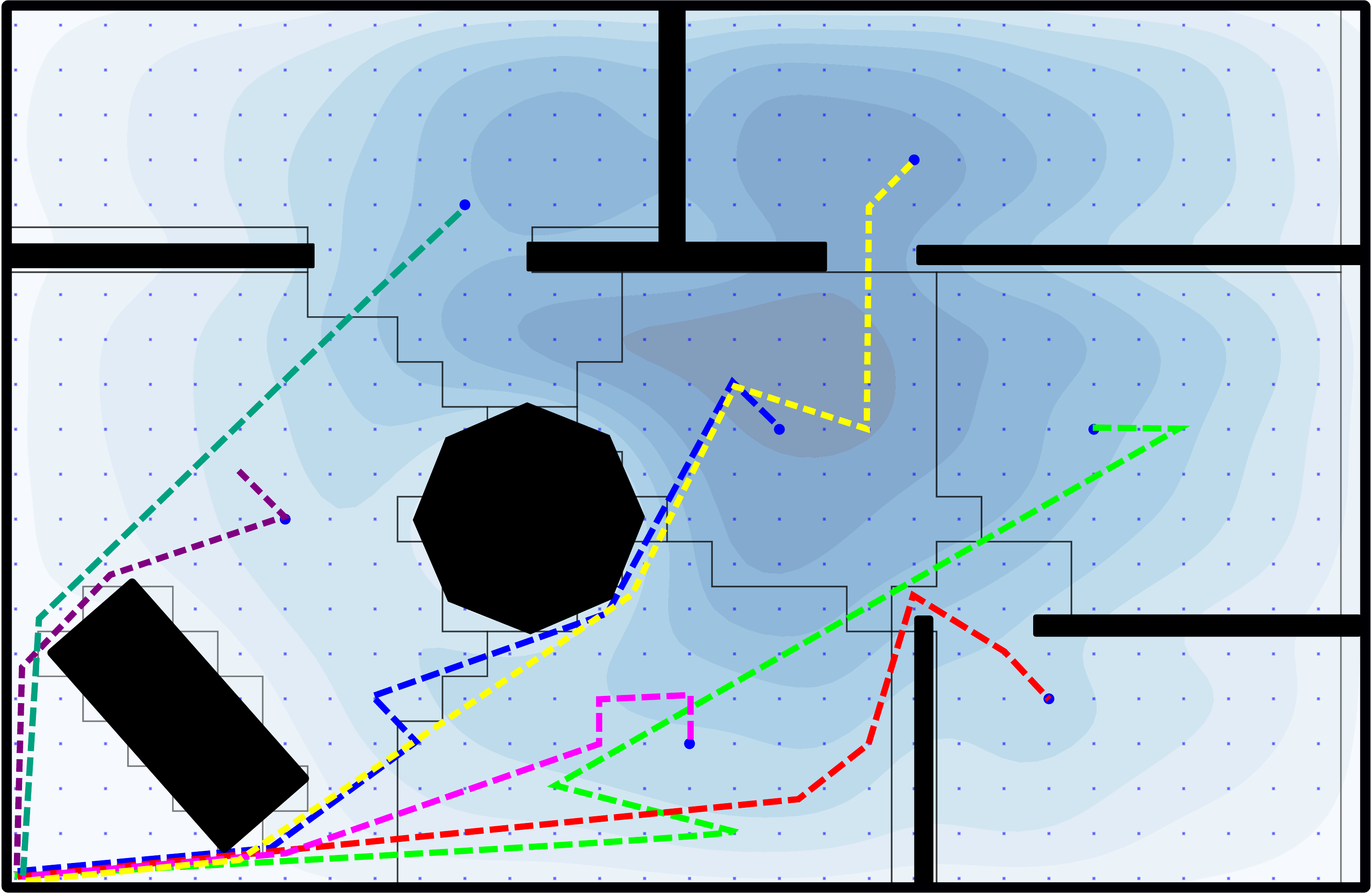}
    \caption{Robot movements in a non-convex environment using the proposed distributed algorithm}
    \label{fig:non_convex_paths}
\end{figure}

\subsection{Examples of Different Local Solutions}
In this section, we provide two examples for the algorithms in~\cite{yun2014distributed} and~\cite{breitenmoser2010voronoi}. We illustrate the locally optimal solution reached using these algorithms, and the local moves considered in the proposed algorithm which helps escaping these sub-optimal solutions.

\begin{figure}
    \begin{subfigure}[t]{0.48\linewidth}
    \begin{center}
    \includegraphics[angle=-90, width=\linewidth]{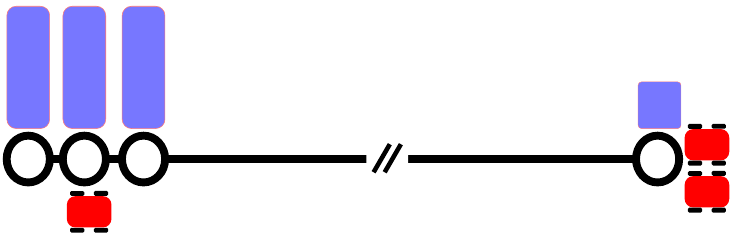}
    \caption{A test environment for algorithm in~\cite{yun2014distributed}}
    \label{fig:discretization-counter-a}
    \end{center}
    
    \end{subfigure}
    \begin{subfigure}[t]{0.48\linewidth}
    \begin{center}
    \includegraphics[angle=-90, width=\linewidth]{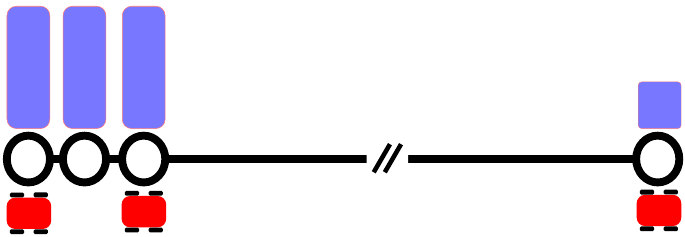}
    \caption{Final configuration with proposed algorithm}
    \label{fig:discretization-counter-b}
    \end{center}
    \end{subfigure}
    \caption{A discrete test environment and the final configurations of algorithm in~\cite{yun2014distributed} and the proposed algorithm}
    \label{fig:discretization-counter}
\end{figure}

Consider a discrete coverage problem with $4$ vertices and $3$ robots initialized at the configuration shown in Figure~\ref{fig:discretization-counter-a}. The bars on the vertices of the graph represents the weight of the vertices. By the communication model in~\cite{yun2014distributed}, all the robots are neighbours of each other. The local move in algorithm in~\cite{yun2014distributed} moves the robots inside their partitions if the move improves the sensing cost of its partition and the neighbouring partitions. Note that the initialized configuration of the robots is a locally optimal solution for the algorithm in~\cite{yun2014distributed}. However, in the proposed algorithm the robots will improve on current configuration with performing single-hop move type II. Figure~\ref{fig:discretization-counter-b} shows the final configuration with the proposed algorithm.

Figure~\ref{fig:schwager_paper_counter} shows a continuous environment with a single robot. The sensing cost of an event is a function of the geodesic distance from the robot. The high-level idea of the algorithm in~\cite{breitenmoser2010voronoi} is to find the centroid in the environment without obstacles (see $t_{virt}$ in Figure~\ref{fig:schwager_paper_counter}) and if the centroid is inside an obstacle, then the algorithm projects the centroid to a face of the obstacle (see $t_{real}$ in Figure~\ref{fig:schwager_paper_counter}) and moves the robot towards $t_{real}$. Observe that in the scenarios where the sensing cost is a function of the length of the shortest path between the robot and the event location, the projection of the centroid may result in sub-optimal solutions. However, the proposed algorithm avoids these scenarios by solving the coverage problem on a discrete representation of the environment.

\begin{figure}
    \centering
    \includegraphics[width=.8\textwidth]{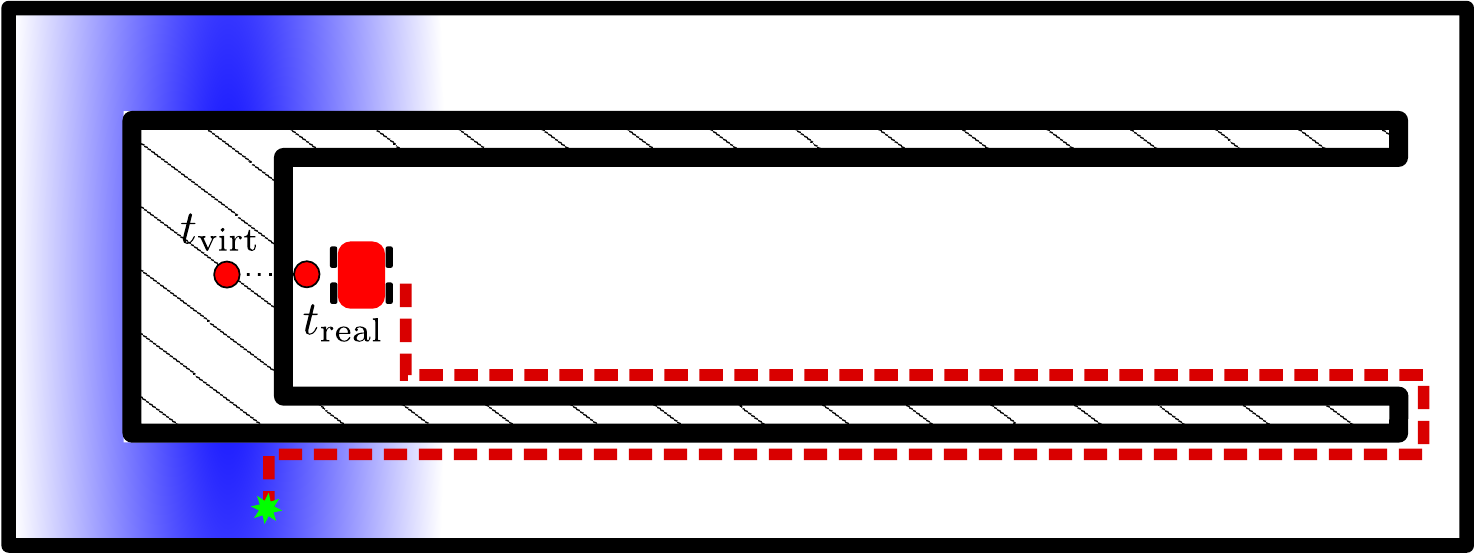}
    \caption{A test environment for Algorithm~\cite{breitenmoser2010voronoi}}
    \label{fig:schwager_paper_counter}
\end{figure}

\section{Conclusion}
This paper considered the multi-robot coverage problem in convex and non-convex environments. A connection is established between the solution quality of the continuous coverage problem and the solution to the coverage problem on a discrete representation of the environment. We also propose the first distributed approximation algorithm for the coverage problem in discrete and continuous environments and provide bound on the quality of the solution. We also characterize the run-time and communication complexity of the proposed algorithm.


\begin{thebibliography}{10}
\providecommand{\url}[1]{{#1}}
\providecommand{\urlprefix}{URL }
\expandafter\ifx\csname urlstyle\endcsname\relax
  \providecommand{\doi}[1]{DOI~\discretionary{}{}{}#1}\else
  \providecommand{\doi}{DOI~\discretionary{}{}{}\begingroup
  \urlstyle{rm}\Url}\fi

\bibitem{ahmadian2013local}
Ahmadian, S., Friggstad, Z., Swamy, C.: Local-search based approximation
  algorithms for mobile facility location problems.
\newblock In: Proceedings of the 24th annual ACM-SIAM symposium on Discrete
  algorithms, pp. 1607--1621. SIAM (2013)

\bibitem{alitappeh2017multi}
Alitappeh, R.J., Pereira, G.A., Ara{\'u}jo, A.R., Pimenta, L.C.: Multi-robot
  deployment using topological maps.
\newblock Journal of Intelligent \& Robotic Systems \textbf{86}(3-4), 641--661
  (2017)

\bibitem{arya2004local}
Arya, V., Garg, N., Khandekar, R., Meyerson, A., Munagala, K., Pandit, V.:
  Local search heuristics for k-median and facility location problems.
\newblock SIAM Journal on computing \textbf{33}(3), 544--562 (2004)

\bibitem{balcan2013distributed}
Balcan, M.F.F., Ehrlich, S., Liang, Y.: Distributed $ k $-means and $ k
  $-median clustering on general topologies.
\newblock In: Advances in Neural Information Processing Systems, pp. 1995--2003
  (2013)

\bibitem{bertsimas1991stochastic}
Bertsimas, D.J., Van~Ryzin, G.: A stochastic and dynamic vehicle routing
  problem in the {E}uclidean plane.
\newblock Operations Research \textbf{39}(4), 601--615 (1991)

\bibitem{bhattacharya2013distributed}
Bhattacharya, S., Michael, N., Kumar, V.: Distributed coverage and exploration
  in unknown non-convex environments.
\newblock In: Distributed Autonomous Robotic Systems, pp. 61--75. Springer
  (2013)

\bibitem{breitenmoser2010voronoi}
Breitenmoser, A., Schwager, M., Metzger, J.C., Siegwart, R., Rus, D.: Voronoi
  coverage of non-convex environments with a group of networked robots.
\newblock In: IEEE International Conference on Robotics and Automation, pp.
  4982--4989 (2010)

\bibitem{bullo2011dynamic}
Bullo, F., Frazzoli, E., Pavone, M., Savla, K., Smith, S.L.: Dynamic vehicle
  routing for robotic systems.
\newblock Proceedings of the IEEE \textbf{99}(9), 1482--1504 (2011)

\bibitem{caicedo2008coverage}
Caicedo-N{\'u}nez, C.H., Zefran, M.: A coverage algorithm for a class of
  non-convex regions.
\newblock In: IEEE Conference on Decision and Control, pp. 4244--4249 (2008)

\bibitem{caicedo2008performing}
Caicedo-Nunez, C.H., Zefran, M.: Performing coverage on nonconvex domains.
\newblock In: IEEE International Conference on Control Applications, pp.
  1019--1024 (2008)

\bibitem{cortes2004coverage}
Cortes, J., Martinez, S., Karatas, T., Bullo, F.: Coverage control for mobile
  sensing networks.
\newblock IEEE Transactions on Robotics and Automation \textbf{20}(2), 243--255
  (2004)

\bibitem{curtin1993autonomous}
Curtin, T.B., Bellingham, J.G., Catipovic, J., Webb, D.: Autonomous
  oceanographic sampling networks.
\newblock Oceanography \textbf{6}(3), 86--94 (1993)

\bibitem{durham2012discrete}
Durham, J.W., Carli, R., Frasca, P., Bullo, F.: Discrete partitioning and
  coverage control for gossiping robots.
\newblock IEEE Transactions on Robotics \textbf{28}(2), 364--378 (2012)

\bibitem{jain2001approximation}
Jain, K., Vazirani, V.V.: Approximation algorithms for metric facility location
  and k-median problems using the primal-dual schema and lagrangian relaxation.
\newblock Journal of the ACM (JACM) \textbf{48}(2), 274--296 (2001)

\bibitem{kantaros2014visibility}
Kantaros, Y., Thanou, M., Tzes, A.: Visibility-oriented coverage control of
  mobile robotic networks on non-convex regions.
\newblock In: IEEE International Conference on Robotics and Automation, pp.
  1126--1131 (2014)

\bibitem{lavalle2006planning}
LaValle, S.M.: Planning algorithms.
\newblock Cambridge university press (2006)

\bibitem{lemaire2004distributed}
Lemaire, T., Alami, R., Lacroix, S.: A distributed tasks allocation scheme in
  multi-uav context.
\newblock In: IEEE International Conference on Robotics and Automation, vol.~4,
  pp. 3622--3627 (2004)

\bibitem{li2016approximating}
Li, S., Svensson, O.: Approximating k-median via pseudo-approximation.
\newblock SIAM Journal on Computing \textbf{45}(2), 530--547 (2016)

\bibitem{mahboubi2012distributed}
Mahboubi, H., Sharifi, F., Aghdam, A.G., Zhang, Y.: Distributed coordination of
  multi-agent systems for coverage problem in presence of obstacles.
\newblock In: IEEE American Control Conference, pp. 5252--5257 (2012)

\bibitem{meguerdichian2001exposure}
Meguerdichian, S., Koushanfar, F., Qu, G., Potkonjak, M.: Exposure in wireless
  ad-hoc sensor networks.
\newblock In: Proceedings of the 7th Annual International Conference on Mobile
  Computing and Networking, pp. 139--150. ACM (2001)

\bibitem{sadeghi2018re}
Sadeghi, A., Smith, S.L.: Re-deployment algorithms for multiple service robots
  to optimize task response.
\newblock In: 2018 IEEE International Conference on Robotics and Automation
  (ICRA), pp. 2356--2363. IEEE (2018)

\bibitem{sadeghi2019coverage}
Sadeghi, A., Smith, S.L.: Coverage control for multiple event types with
  heterogeneous robots.
\newblock In: IEEE International Conference on Robotics and Automation, pp.
  3377--3383 (2019)

\bibitem{santos2018coverage}
Santos, M., Diaz-Mercado, Y., Egerstedt, M.: Coverage control for multirobot
  teams with heterogeneous sensing capabilities.
\newblock IEEE Robotics and Automation Letters \textbf{3}(2), 919--925 (2018)

\bibitem{shmoys2000approximation}
Shmoys, D.B.: Approximation algorithms for facility location problems.
\newblock In: International Workshop on Approximation Algorithms for
  Combinatorial Optimization, pp. 27--32. Springer (2000)

\bibitem{thanou2013distributed}
Thanou, M., Stergiopoulos, Y., Tzes, A.: Distributed coverage using geodesic
  metric for non-convex environments.
\newblock In: IEEE International Conference on Robotics and Automation, pp.
  933--938 (2013)

\bibitem{wang2011coverage}
Wang, B.: Coverage problems in sensor networks: A survey.
\newblock ACM Computing Surveys (CSUR) \textbf{43}(4), 1--53 (2011)

\bibitem{yun2014distributed}
Yun, S.k., Rus, D.: Distributed coverage with mobile robots on a graph:
  locational optimization and equal-mass partitioning.
\newblock Robotica \textbf{32}(2), 257--277 (2014)

\end{thebibliography}
\end{document}